\documentclass[preprint,12pt]{elsarticle}

\usepackage{graphicx}
\usepackage{subfigure}

\usepackage[utf8]{inputenc}
\usepackage{amsmath}
\usepackage{amsthm}
\usepackage{amsfonts}
\usepackage{epsfig}
\usepackage{psfrag}
\usepackage{pstricks}
\usepackage{algorithm}

\usepackage{amssymb}
\usepackage{bm}

\def \div {\mathrm{div}}
\usepackage{fancyhdr}
\usepackage{color}

\definecolor{turquoise}{cmyk}{0.65,0,0.1,0.1}
\definecolor{purple}{rgb}{0.65,0,0.65}
\definecolor{dark_green}{rgb}{0, 0.5, 0}
\definecolor{orange}{rgb}{0.8, 0.2, 0.2}

\newtheorem{prop}{Proposition}
\newtheorem{lem}{Lemma}

\newcommand{\RevisionBlack}[1]{{\color{black}#1}}

\biboptions{comma,compress}

\journal{COMPUTER-AIDED DESIGN}

\begin{document}

\begin{frontmatter}



\title{Iso-level tool path planning for free-form surfaces}


\author[ustc,sia]{Qiang Zou}
\ead{john.qiangzou@gmail.com}

\author[ustc]{Juyong Zhang \corref{cor2}}
\ead{juyong@ustc.edu.cn}

\author[epfl]{Bailin Deng}
\ead{bailin.deng@epfl.ch}

\author[sia]{Jibin Zhao}
\ead{jbzhao@sia.cn}

\cortext[cor2]{Corresponding author. Tel.: +86-551-63600673.}

\address[ustc]{School of Mathematical Sciences, University of Science and Technology of China, Anhui, 230026, China;}
\address[epfl]{Computer Graphics and Geometry Laboratory, \'Ecole Polytechnique F\'ed\'erale de Lausanne, Lausanne, CH-1015, Switzerland;}
\address[sia]{Shenyang Institute of Automation, Chinese Academy of Sciences, Liaoning, 110016, China.}


\begin{abstract}
The aim of tool path planning is to maximize the efficiency against some given precision criteria. In practice, scallop height should be kept constant to avoid unnecessary cutting, while the tool path should be smooth enough to maintain a high feed rate. However, iso-scallop and smoothness often conflict with each other. Existing methods smooth iso-scallop paths one-by-one, which makes the final tool path far from being globally optimal. This paper proposes a new framework for tool path optimization. It views a family of iso-level curves of a scalar function defined over the surface as tool path so that desired tool path can be generated by finding the function that minimizes certain energy functional and different objectives can be considered simultaneously. We use the framework to plan globally optimal tool path with respect to iso-scallop and smoothness. The energy functionals for planning iso-scallop, smoothness, and optimal tool path are respectively derived, and the path topology is studied too. Experimental results are given to show effectiveness of the proposed methods.

\end{abstract}

\begin{keyword}
iso-level tool path \sep globally optimal \sep PDE \sep iso-scallop \sep smooth

\end{keyword}

\end{frontmatter}


\section{Introduction}
\label{sec:1}
The terminology ``tool path'' refers to a specified trajectory along which machine tools move their ends (i.e., cutter and table) to form desired surfaces. The automatic generation of such trajectories are of central importance in modern CAD/CAM systems. There are two fundamental criteria, i.e., precision and efficiency, for automatic tool path generation. Precision means the error of approximating a surface with a family of curves, and approximating a curve with a family of segments or arcs. Efficiency concerns the time of machining along the tool path. The aim of tool path planning is to maximize the efficiency under the given precision criteria. In this paper, we propose a method, which can take these two criteria into consideration together, to generate globally optimal tool paths.

\subsection{Related works}
\label{sec:1.1}
For a given precision tolerance (i.e., the scallop height and chord deviation), the tool path is always supposed to be as smooth and short as possible. In this paper, the smoothness of tool path is measured by its curvature in the 3D space. If the tool path is smooth enough, there is less repeated acceleration/deceleration, which makes it possible to maintain a high feed rate. Meanwhile, the shorter the tool path is, the less time the machining takes. Theoretically, tool paths following the direction of maximum machining strip width are the shortest in total length, since they maximize material removal. But such strategy often leads to irregular tool paths which are neither direction/contour parallel nor spiral, as shown in \cite{kim2002toolpath, kumazawa2012generating}. Therefore, in practice, a weaker condition that the tool path has no unnecessary (also called redundant) cutting is adopted. To achieve this, the scallop height should be kept constant along the path. Hence, tool paths with iso-scallop and smooth properties are preferable.

Last decade has seen a great deal of literature on tool path planning for free-form surfaces, such as iso-parametric method \cite{loney1987nc, yuwen2006iso, Qiang2013iso}, iso-planar method \cite{huang1994non, ding2003adaptive, feng2005iso}, iso-scallop method \cite{suresh1994constant, koren1996efficient, sarma1997geometry, feng2002constant, yoon2005fast, kim2007constant, li2012power}, iso-phote method \cite{han1999iso} and C-space method \cite{choi2001c}, to name a few. Surveys of much more work about tool path planning research can be found in \cite{dragomatz1997classified, lasemi2010recent}. Since we aim at optimal tool paths with respect to iso-scallop and smoothness, we put special interest in the iso-scallop method, which means the height of the points at the scallop curves remains as high as a given value so that the tool path has no unnecessary cutting. Conventionally, constant scallop height is obtained by varying the offset magnitude along each path. A mathematical method for generating iso-scallop tool paths following such strategy was first proposed by Suresh et al. \cite{suresh1994constant}. Afterwards, methods to improve the computing efficiency \cite{koren1996efficient, yoon2005fast} and accuracy \cite{sarma1997geometry, feng2002constant, li2012power} were proposed. In 2007, Kim \cite{kim2007constant} reformulated the iso-scallop tool path as geodesic parallel curves on the design surface by defining a new Riemannian metric.

Despite the non-redundance property, tool paths of constant scallop height tend to have sharp corners, as illustrated in Fig.~\ref{fig:iso-scallop-toolpath}, which implies that smoothness and iso-scallop requirements often conflict with each other. And the tradeoff between them is a major concern in tool path planning. A widely adopted solution to thid is postprocessing: first a new path with constant scallop is generated by varying the offset magnitude along current path, then it is smoothed by replacing its corners with circular arcs \cite{pateloup2004corner, pateloup2010bspline}. An alternative is to employ the level set method to offset the paths while keeping them smooth \cite{dhanik2010contour}. Similar to the image segmentation method proposed by Paragios et al. \cite{paragios2002geodesic}, a curvature term can be added into the evolution equation so that points of higher curvature are offset less while those of lower curvature are offset more. However, on one hand, for a path subject to desired precision, the modification would introduce error. On the other hand, if the path is offset less than the desired tolerance to avoid such error, then hardly can we choose a proper offset magnitude since we usually do not have an overall picture of the tool path. For example, in \cite{dhanik2010contour}, because a curvature item is introduced into the normal velocity, it is still unknown how to choose a proper evolution step that determines the distance between neighboring paths. In fact, such local modification is in general unable to gain a globally optimal tool path, since it cannot take the ungenerated paths into account when operating on (or optimizing) one path. All previous offset based methods generate tool paths one-by-one, and thus inherit the drawback of non-optimality.

There also exist some efforts to generate smooth tool paths without considering the overlapping between neighbor machining strips (i.e., the iso-scallop condition). Generally, such methods are based on the Laplacian. For example, Bieterman and Sandstrom \cite{bieterman2003curvilinear} proposed a Laplacian based contour parallel tool path generation method by selecting the level sets of a harmonic function defined over a pocket as the tool path. But how to choose the level sets for it still remains an open problem, namely there is no formula for path interval calculation so far. Similarly, Chuang and Yang \cite{chuang2007laplace} combined the Laplacian method and iso-parametric method to generate tool paths for pockets with complex topology, i.e., complex boundaries and islands. However, the smoothness of the tool path cannot be guaranteed through Laplacian energy as small Laplacian value does not necessarily mean small curvature of the level set curves. And solving a Laplace equation over a surface can only generate a unique and uncontrollable scalar function (scaling has no impact on the shape of tool paths).  Another drawback of the Laplacian based approach is the severe overlapping between machining strips of neighbor paths, especially for paths near the boundary, which results in too much redundant machining.


\subsection{Our approach}
\label{sec:1.2}
In this paper, we aim to plan optimal tool path regarding iso-scallop and smoothness. We propose a framework that is able to obtain a globally optimal tool path by considering several objectives together. The tool path is represented as a family of level set curves from a scalar function defined over the surface, and our method computes an optimal scalar function by solving a single optimization problem, instead of generating the curves one-by-one. We refer to the level sets as \emph{iso-level curves}, and the proposed tool path planning method as \emph{iso-level method}, in order to be consistent with other terminologies in the literature such as iso-parametric, iso-planar, iso-scallop and iso-phote.

As the tool path is represented by the iso-level curves of some optimized scalar function, desired properties of the tool path are encoded into the properties of the scalar function. In this work, we give the details of how to control the scalar function so that the desired tool path, e.g., iso-scallop tool path, can be generated. We first propose an iso-scallop condition for the target function, which shapes two neighboring iso-level curves to be iso-scallop. Then we propose a smoothness objective. Finally we combine them together to form the objective energy functional so that its minimizer corresponds to an optimal tool path with respect to iso-scallop and smoothness. To the best of our knowledge, this paper is the first work where these formulas are given, through which interval between iso-level curves and their smoothness can be controlled globally. The minimizer of the iso-scallop objective can not only be exploited to plan tool path of constant scallop, but also has an interesting machining meaning, namely, the level increment of two neighbor iso-level curves equals to the square root of scallop height. In addition, the optimal scalar function can be reused to generated tool path of different scallop height tolerances.

Compared with existing tool path generation methods, the proposed method solves the tool path planing problem in a global optimization manner. Besides, the proposed iso-level tool path planning method can free us from the tedious post-processing step for self-intersection and disjunction, which will be demonstrated in more details in Section \ref{sec:2.4}. In addition, since the scalar function is defined all over the surface, the model is completely covered by the iso-level curves, i.e., there are no regions that are not machined, as opposed to the offset based methods (illustrated in Fig. \ref{fig:iso-scallop-toolpath}). Our optimization framework can also be easily extended to include other objectives, such as tool wear, machine kinematics and dynamics.

\begin{figure}[htbp]
\centering
\hfill
\subfigure[]{
\includegraphics[width=0.37\textwidth]{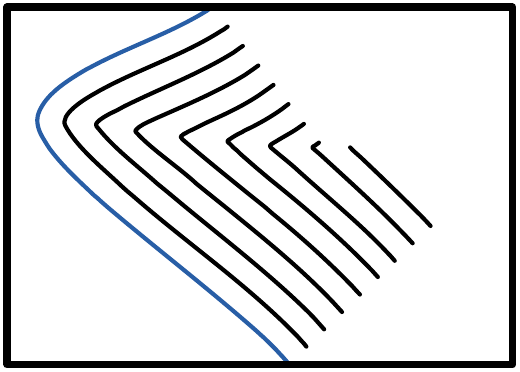}
\label{fig:1a}
}
\hfill
\subfigure[]{
\includegraphics[width=0.36\textwidth]{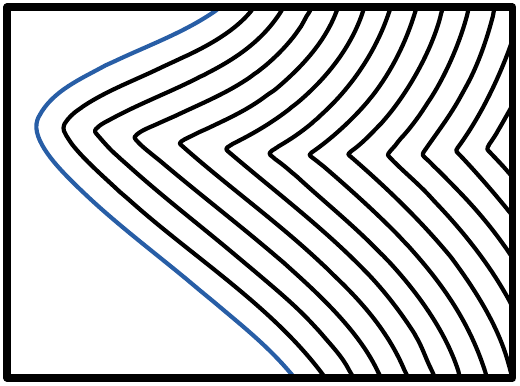}
\label{fig:1b}
}
\hspace*{\fill}
\caption{Iso-scallop tool path. (a) offset based iso-scallop tool path; (b) scale function based iso-scallop tool path.}
\label{fig:iso-scallop-toolpath}
\end{figure}

The remainder of this paper is organized as follows: Section~\ref{sec:2} describes the optimization models for iso-level method, including iso-scallop tool path generation (Section~\ref{sec:2.1}), smooth tool path generation (Section~\ref{sec:2.2}), optimal tool path generation (Section~\ref{sec:2.3}), followed by a discussion on tool path topology (Section~\ref{sec:2.4}). In Section~\ref{sec:3}, we present the numerical solution to the optimization models. Section~\ref{sec:4} summarizes the overall procedures for planning iso-level tool paths. Section~\ref{sec:5} shows the experimental results. Finally, we conclude the whole paper in Section \ref{sec:6}.

\begin{figure*}[htb]
\centering
\hfill
\subfigure[]{
\includegraphics[width=0.45\textwidth, origin=bl]{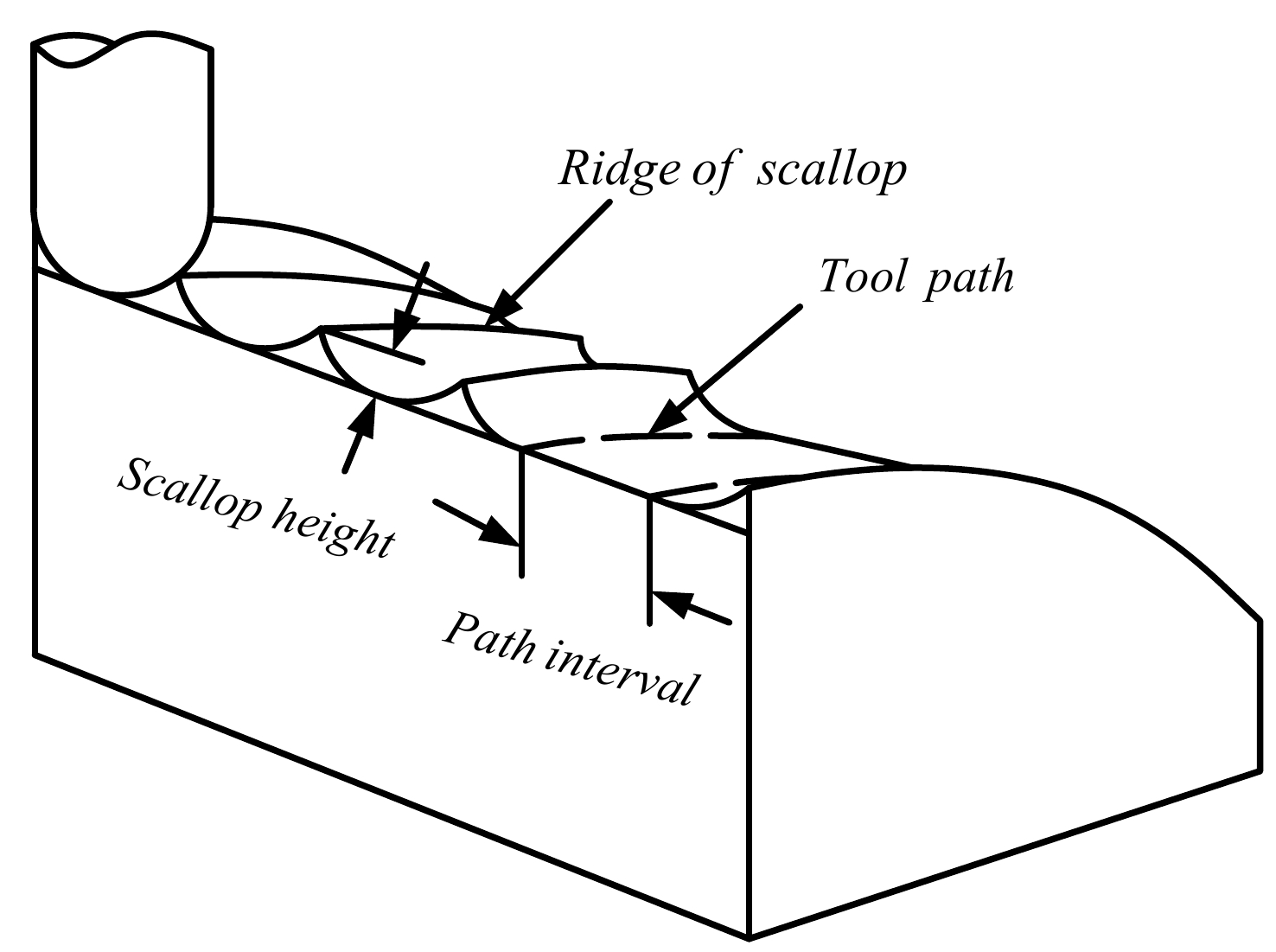}
\label{fig:path-parameters}
}
\hfill
\subfigure[]{
\includegraphics[width=0.45\textwidth]{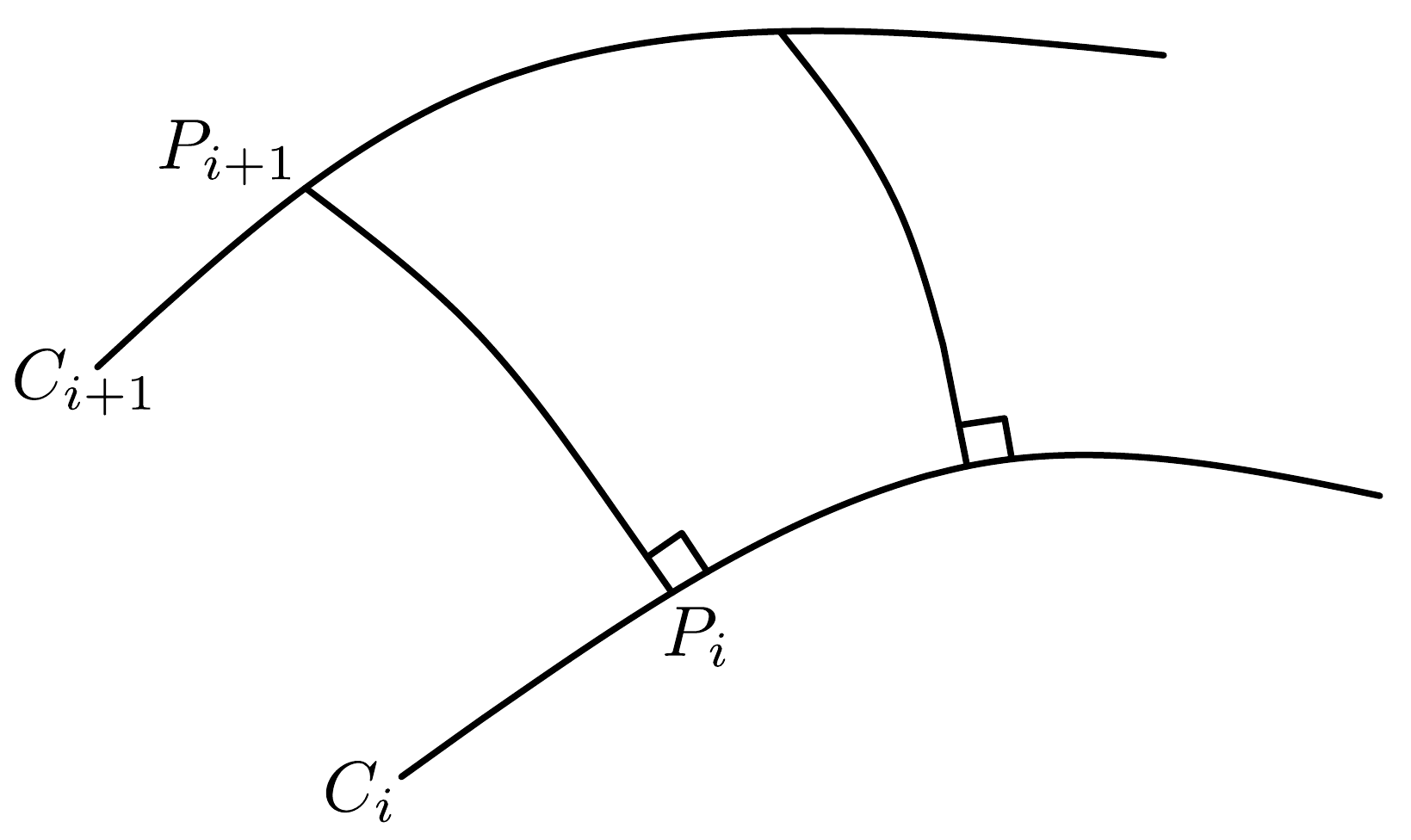}
\label{fig:path-interval}
}
\hspace*{\fill}
\caption{Illustrations of path parameters. (a) scallop height and path interval; (b) path interval.}\label{path-parameter-illustration}
\end{figure*}

\section{Optimal iso-level tool path}
\label{sec:2}
Consider a surface $S$ embedded in ${\mathbb{R}}^{3}$ and a scalar function $\varphi: S\to \mathbb{R}$ defined over it. The curves on $S$ which correspond to a set of values $\{l_i\}_{i=1}^n$ bounded by the range of the scalar function are selected as tool path for the surface. There are two problems to concern when generating tool path following this strategy: the design of $\varphi$ and the mathematical method for determining $\{l_i\}$. In this section, we describe our solution to them, and demonstrate how to plan iso-level tool paths.

\subsection{Iso-scallop tool path generation}
\label{sec:2.1}
In general, a tool path is discretized as a family of curves on the surface. Scallop refers to the remaining material that is generated when the cutter sweeps along two neighbor paths, which results in deviation between the machined surface and the design surface. Generally, we use the height from the points at the ridge of the scallop to the design surface to quantify such error, as illustrated in Fig. \ref{fig:path-parameters}. On one hand, the closer the two neighboring curves are, the lower the scallop height becomes. On the other hand, closer curves may lead to longer path and time to machine the whole surface. Therefore, the iso-scallop method generates tool paths with the scallop height as high as a specific tolerance in order to avoid redundant machining and achieves higher efficiency. The scallop height is determined by the interval between two neighbor paths and they are related by the following formula \cite{koren1996efficient}
\begin{equation}\label{equ:1}
h=\frac{{{\kappa }_{s}}+{{\kappa }_{c}}}{8}{{w}^{2}}+{\mathrm O}\left( {{w}^{3}} \right),
\end{equation}
where $w$ denotes the interval ${p_i}{p_{i+1}}$, $h$ is the scallop height, $\kappa_s$  is the normal curvature along the direction normal to path $C_i$, as shown in Fig. \ref{fig:path-interval}, and $\kappa_c$ is the curvature of the cutter.

Let $C_i, C_{i+1}$ be the iso-level curves $\left\{ p\in S\ |\ \varphi (p)={{l}_{i}} \right\}$ and $\left\{ p\in S\ |\ \varphi (p)={{l}_{i+1}} \right\}$, respectively. Then appeal to the Taylor's theorem, we have
\begin{equation}\label{equ:2}
{{l}_{i+1}}-{{l}_{i}}={{\left( \nabla \varphi  \right)}^{T}}\left( {{p}_{i+1}}-{{p}_{i}} \right)+{\mathrm O}\left( {{\left\| {{p}_{i+1}}-{{p}_{i}} \right\|}^{2}} \right).
\end{equation}
The gradient $\nabla \varphi$ is a vector in the tangent plane of the surface at point $p_i$, and normal to the path. Therefore, the expression can be rewritten as
\begin{equation}\label{equ:3}
\left| {{l}_{i+1}}-{{l}_{i}} \right|=\left\| \nabla \varphi  \right\|\cdot \left\| {{p}_{i+1}}-{{p}_{i}} \right\|+{\mathrm O}\left( {{\left\| {{p}_{i+1}}-{{p}_{i}} \right\|}^{2}} \right).
\end{equation}
Then
\begin{equation}\label{equ:4}
\left\| \nabla \varphi \right\|=\lim_{\left\| {{p}_{i+1}}-{{p}_{i}} \right\| \rightarrow 0} \frac{\left| {{l}_{i+1}}-{{l}_{i}} \right|}{\left\| {{p}_{i+1}}-{{p}_{i}} \right\|}.
\end{equation}
If the level increment $\left| {{l}_{i+1}}-{{l}_{i}} \right|$ of the scalar funciton is endowed with a machining meaning by letting it equal to the square root of scallop height, Eq.~\eqref{equ:4} will be
\begin{equation}\label{equ:5}
\left\| \nabla \varphi  \right\|=\sqrt{\frac{{{\kappa }_{s}}+{{\kappa }_{c}}}{8}},
\end{equation}
and the scallop height between two neighbor iso-level curves of $\varphi$ will be constant and equal to the square of the increment. This can be easily verified by substituting Eq.~\eqref{equ:1} into Eq.~\eqref{equ:4}.

Thus an iso-scallop tool path can be generated by finding a scalar function satisfying Eq.~\eqref{equ:5}. And we obtain such $\varphi$ by solving a nonlinear least square problem
\begin{equation}\label{equ:iso-scallop-energy}
\RevisionBlack{\min_{\varphi} \ \ {{E}_{w}(\varphi)}=\int_{\mathcal{S}}{{{\left( \left\| \nabla \varphi  \right\|-\sqrt{\frac{{{\kappa }_{s}}+{{\kappa }_{c}}}{8}} \right)}^{2}}d\mathcal{S}},}
\end{equation}
where $\kappa_c$ is a user input and $\kappa_s$ is computed by

\begin{equation}\label{equ:7}
\RevisionBlack{{{\kappa }_{s}} = {{\left( \frac{\nabla\varphi}{\|\nabla\varphi\|} \right)}^{T}}T\left(\frac{\nabla\varphi}{\|\nabla\varphi\|}\right),}
\end{equation}
with $T$ denoting the curvature tensor (see \cite{Carmo76, rusinkiewicz2004estimating} for its definition and numerical computation).

Finally, the iso-level curves corresponding to level values $\left\{ i\sqrt{h} \right\}_{i=1}^{n}$ are iso-scallop tool path with constant height $h$, that is, level increments between neighboring iso-level curves all equal to $\sqrt{h}$. Thus, a large $\sqrt{h}$ can generate a tool path for rough machining, and a small increment for finish machining. The novelty here is that they share the same scalar function. We refer to this as multiresolution property.

\subsection{Smooth tool path generation}
\label{sec:2.2}

As explained in the introduction section, a smooth tool path is preferred as we can get a nearly constant feed rate along it.
For a curve in 3D space, its curvature measures how much it bends at a given point. This is quantified by the norm of its second derivative with respect to arc-length parameter, which measures the rate at which the unit tangent turns along the curve~\cite{Carmo76}. It is the very metric to measure the smoothness of the curve.

As the tool path is embedded on the design surface, its second derivative with respect to arc-length parameter can be decomposed into two components, one tangent to the surface and the other normal to the surface (see Fig.~\ref{fig:curvature}) \citep{Carmo76}. The norms of these components are called the geodesic curvature and the normal curvature respectively, and they are related to the curve curvature by
\begin{equation}\label{equ:smooth-term}
\kappa^{2} = \kappa_{g}^{2} + \kappa_{n}^{2},
\end{equation}
where $\kappa$ is the curve curvature, and $\kappa_{g}, \kappa_{n}$ are the geodesic curvature and the normal curvature respectively.

For an iso-level curve $\phi = const$, its normal curvature can be computed by
\begin{equation}
\label{eq-normalCurvature1}
\begin{array}{ccl}
{\kappa }_{n} &=& \left(n\times\frac{\nabla \varphi}{\|\nabla \varphi\|} \right)^{T} T \left( n\times\frac{\nabla \varphi}{\|\nabla \varphi\|} \right) \\
&=& \left( A\frac{\nabla\varphi}{\|\nabla\varphi\|})^{T}T(A\frac{\nabla\varphi}{\|\nabla\varphi\|} \right) \\
&=& \left( \frac{\nabla\varphi}{\|\nabla\varphi\|})^{T}T^{'}(\frac{\nabla\varphi}{\|\nabla\varphi\|} \right),
\end{array}
\end{equation}
where $T$ is the curvature tensor, $n=\left( {{n}_{x}},{{n}_{y}},{{n}_{z}} \right)^T$ is the unit normal vector of the surface, and
\begin{equation}
\label{eq-cross}
A = \left[ \begin{array}{ccc}
0 & -n_{z} & n_{y} \\
n_{z} & 0 & -n_{x} \\
-n_{y} & n_{x} & 0 \end{array} \right], \quad T^{'} = A^{T}TA.
\end{equation}
The geodesic curvature can be computed by
\begin{equation}\label{eq-geodesic}
\kappa_{g} = \div \left( \frac{\nabla \varphi}{\|\nabla \varphi\|} \right),
\end{equation}
where $\div(\cdot)$ is the divergence operator. For a planar curve, its normal curvature is zero and we have
\begin{equation}\label{eq-planar-curvature}
\kappa ={{\kappa }_{g}}=div(\frac{\nabla \varphi }{\|\nabla \varphi \|})\ne div(\nabla \varphi )={{\nabla }^{2}}(\varphi ).
\end{equation}
Therefore, Laplacian can't ensure smoothness of a tool path for pocket milling.

To guarantee the smoothness of all iso-level curves on surface $S$, we define the smoothness energy as
\begin{equation}
\begin{array}{ccl}
{E_{\kappa }}(\varphi )  = \int_{\mathcal{S}}\kappa^{2} d\mathcal{S}
 = \int_{\mathcal{S}}\kappa_{g}^{2}d\mathcal{S} + \int_{\mathcal{S}}\kappa_{n}^{2}d\mathcal{S}.
\end{array}
\label{eq:curveCurvature}
\end{equation}

In Section \ref{sec:2.1}, we employ the formula $\left| {{l}_{i+1}}-{{l}_{i}} \right| = \sqrt{h}$  to generate iso-level tool path. But for smooth tool path, the following strategy is exploited: First, a certain number of points are sampled from iso-level curve $C_i$; Then the level increment $\left| {{l}_{i+1}}-{{l}_{i}} \right|$ is computed for each point with respect to a given scallop height $h$ using Eq.~\eqref{equ:1} and Eq.~\eqref{equ:3}; Finally, the smallest level increment is chosen to be the level increment between $C_i$ and its next path $C_{i+1}$. This results in level increments of different values as opposed to the iso-scallop method, while the scalar function remains unchanged, i.e., the multiresolution property still holds.
\begin{figure}[htb]
\centering
\hfill
\includegraphics[width=0.37\textwidth]{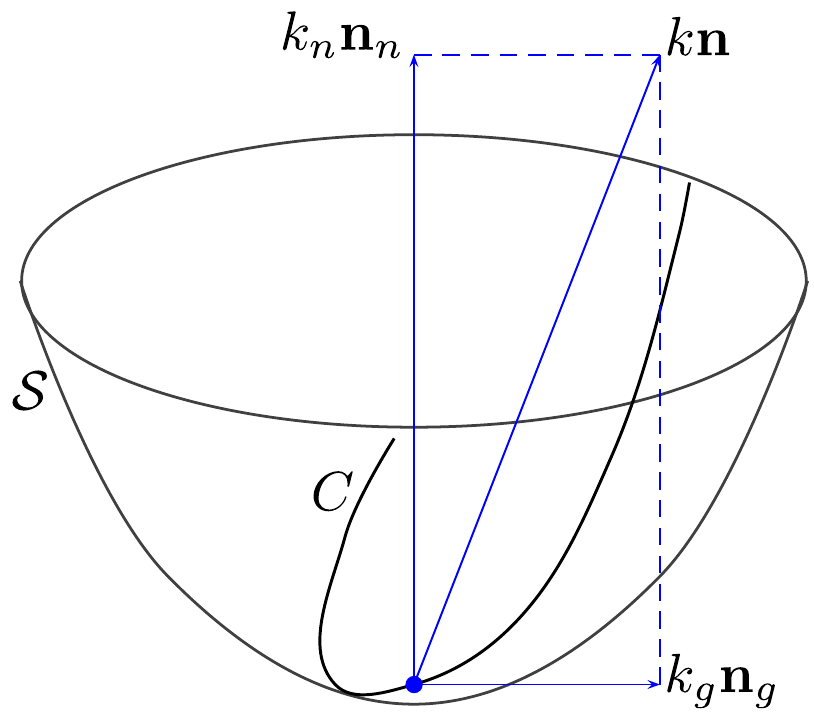}
\hspace*{\fill}
\caption{The curvature vector $k\mathbf{n}$ of curve $C$ on $\mathcal{S}$ has two orthogonal components: the normal curvature vector
$k_{n} \mathbf{n}_n$ and the geodesic curvature vector $k_{g} \mathbf{n}_g$.}\label{fig:curvature}
\end{figure}

\subsection{Optimal tool path generation}
\label{sec:2.3}
The width term Equ.~\eqref{equ:5} and the smoothness term Equ.~\eqref{equ:smooth-term} can control the interval between neighboring paths and smoothness of the paths respectively. Thus an optimal tool path in terms of iso-scallop and smoothness can be obtained by computing $\varphi$ through a nonlinear least square optimization which minimizes a linear combination of the two energies Equ.~\eqref{equ:iso-scallop-energy} and Equ.~\eqref{eq:curveCurvature}
\begin{equation}\label{equ:10}
E(\varphi )=E_w (\varphi) + \lambda E_\kappa (\varphi), 
\end{equation}
where $\lambda$ is a positive weight controlling the trade-off between the two terms. And in order to ensure the tool path is regular (i.e., either contour parallel or direction parallel), we introduce a hard constraint $\left\| \nabla \varphi  \right\| >  0$. The impact of this constraint is demonstrated in Section \ref{sec:2.4}. Then the optimization problem becomes
\begin{align}
	\min_{\varphi} & \int_{S}{{{\left( \left\| \nabla \varphi  \right\|-\sqrt{\frac{{{\kappa }_{s}}+{{\kappa }_{c}}}{8}} \right)}^{2}}+\lambda \left( \kappa _{g}^{2}+\kappa _{n}^{2} \right)}d\mathcal{S} \nonumber \\
	\textrm{s.t.} & \quad \left\| \nabla \varphi  \right\| >  0.
\label{eq:OptimizationModel-continue}
\end{align}

However, because of the smoothness energy, the optimization result may violate Equ.~\eqref{equ:5}, and thus the formula $\left| {{l}_{i+1}}-{{l}_{i}} \right| = \sqrt{h}$ would be invalid. Therefore, we employ the method described in Section \ref{sec:2.2} to select iso-level curves with respect to a certain scallop height tolerance. Note that we can use the same scalar function for planning tool paths of different scallop height tolerances, which again shows the multiresolution property of our approach.

Since different machine tools have different feed rate capability, for those of good capability we can choose a lower weight on smooth term. Thus the freedom of choosing weights ${\lambda }$ provides the possibility of applying the proposed method to various machine tools.

\subsection{Path topology}
\label{sec:2.4}
In this section, we will show that each iso-level curve generated by the proposed method is either a closed loop or a curve segment without self-intersection and disjunction. In addition, this kind of path topology can be exploited to quickly extract iso-level curves.

\begin{lem}
\label{lemma1}
For a given scalar function $\varphi$ defined over a surface $S$, if the norm of its gradient does not vanish anywhere, the endpoints of iso-level curves (if they exist) are on the boundary.
\end{lem}
\begin{proof}
For an interior point $p$, $\| \nabla \varphi \| \neq 0$ implies that along the two directions $\triangle p_1 , \triangle p_2$ orthogonal to $\nabla \varphi$, we have, in a small range, the following expression
\begin{equation}
\varphi (p + \triangle p_i) - \varphi (p) = {(\nabla \varphi)}^T \triangle p_i = 0 \quad \textrm{for}\quad i=1,2.
\end{equation}
Namely, each interior point has exactly two directions sharing the same level value with it. Therefore, the endpoints can only be on the boundary.
\end{proof}

\begin{lem}
\label{lemma2}
For the scalar function $\varphi$, its iso-level curves never intersect with each other and do not have self-intersections.
\end{lem}
\begin{proof}
Since each point corresponds to a unique value, iso-level curves for different values do not intersect with each other. Generally, we have two types of self-intersections, as shown in Fig. \ref{fig:selfIntersection}. The difference between them is that in case (b) the self-intersection is tangential. For case (a), the two curve segments have different tangent directions at the self-intersection point. It is well-known that the gradient direction at a point is orthogonal to the tangent direction of the iso-level curve. Thus the two different tangent directions at the self-intersection point results in a contradiction that there are two different gradient direction at that point. For case (b), we view the self-intersection as two iso-level curves that are of same level value and tangential at the intersection point and separate from each other in its neighborhood. But $\| \nabla \varphi \| \neq 0$ along an iso-level curve implies that iso-level curves near it are of different level values, which again leads to a contradiction.
\end{proof}


\begin{figure}[htbp]
\centering
\hfill
\subfigure[]{
\includegraphics[width=0.37\textwidth]{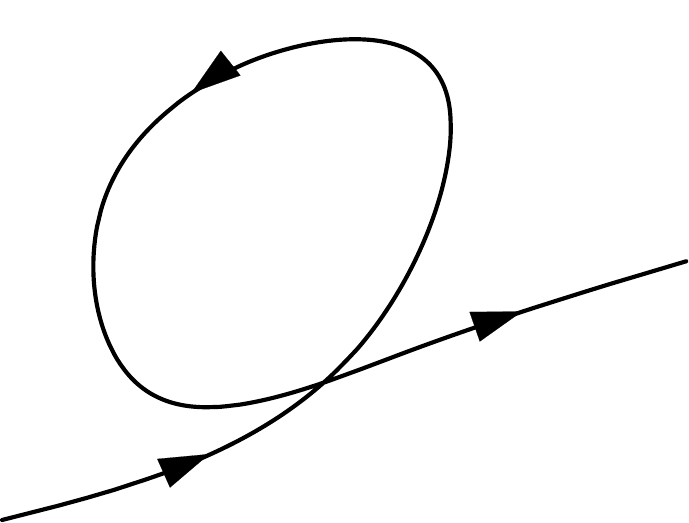}
\label{fig:selfIntersection-a}
}
\hfill
\subfigure[]{
\includegraphics[width=0.27\textwidth]{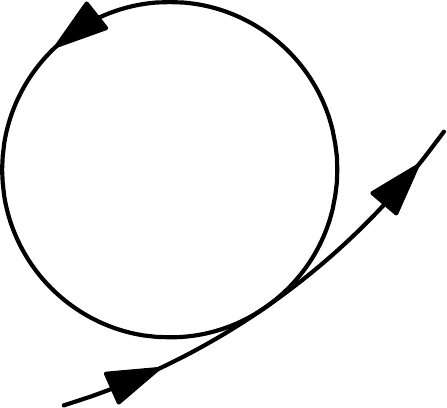}
\label{fig:selfIntersection-b}
}
\hspace*{\fill}
\caption{Illustrations of self-intersection}\label{fig:selfIntersection}
\end{figure}

Note that these properties are employed to extract each desired iso-level curves in the following sections so that traversal is not needed. As Lemma \ref{lemma1} and Lemma \ref{lemma2} show, for any interior point of the surface, it has and only has two directions that share the same level value with the point. Accordingly, we can use a "Seed Growth" like algorithm to find the iso-level curves, namely, if we want to find the iso-level curve of a given level value, say $l$, we can start from an initial edge on which there exists a point whose level value is $l$, and then search through the edge's adjacent triangles (if the point is a vertex of the mesh, i.e., the endpoint of the edge, the adjacent triangles are all the 1-ring triangles) to get exactly two edges containing level value $l$, repeat this procedure and finally the initial point can grow to be an iso-level curve of interest. In addition, it follows immediately from Lemma \ref{lemma1} and Lemma \ref{lemma2} that:
\begin{prop}
\label{proposition}
Each iso-level path generated by the proposed method is either contour parallel or direction parallel and free from self-intersection and disjunction.
\end{prop}

\section{Numerical solution}
\label{sec:3}
Numerically, the iso-level method described in the above section can be applied to any domain with a discrete gradient operator $\nabla$, divergence operator $\div(\cdot)$, and curvature tensor $T$. To solve the optimization models for free-from surfaces, we appeal to the Finite Element Method (FEM), i.e., in this work, we focus on triangular meshes. However, this method can be easily extended to other domains, such as point clouds.

Assume that $M\subset\mathbb{R}^3$ is a compact triangulated surface with no degenerate triangles. Let $N_1(i)$ be the $1$-neighborhood of vertex $v_i$, which is the index set for vertices connecting to $v_i$. Let $D_1(i)$ be the $1$-disk of the vertex $v_i$, which is the index set for triangles containing $v_i$. The dual cell of a vertex $v_i$ is part of its $1$-disk which is more near to $v_i$ than its $N_1(i)$. Fig.~\ref{fig-DMeshCCellCCoeff} (a) shows the dual cell $C_i$ for an interior vertex $v_i$, while Fig.~\ref{fig-DMeshCCellCCoeff} (b) shows the dual cell for a boundary vertex. A function $\varphi$ defined over the triangulated surface $M$ is considered to be a piecewise linear function, such that $\varphi$ reaches value $\varphi_i$ at vertex $v_i$ and is linear within each triangle. Based on these, the energies shown in Equ.~\eqref{equ:iso-scallop-energy}, Equ.~\eqref{eq:curveCurvature}, and Equ.~\eqref{eq:OptimizationModel-continue} are computed by integrating the width term and smooth term over the whole mesh domain, while the mesh domain can be decomposed into a set of triangles or a set of dual cells. To compute the width term and the smooth term on a mesh, we need to discretize the gradient, the divergence, and the curvature tensor, which we will describe briefly, since they are basic in FEM.

\begin{figure}[htbp]
\centering
\hfill
\subfigure[]{
\includegraphics[width=0.28\textwidth]{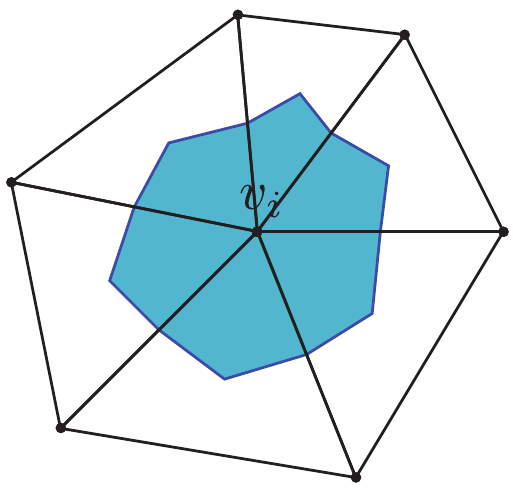}
}
\hfill
\subfigure[]{
\includegraphics[width=0.28\textwidth]{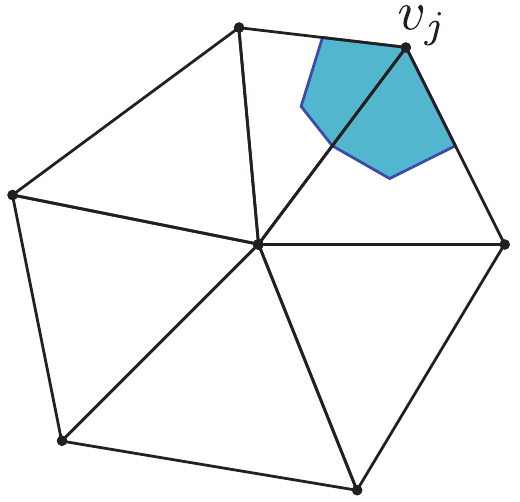}
}
\hspace*{\fill}
\caption{ Dual cells for triangular meshes. (a) dual cell of an interior vertex; (b) dual cell of a boundary vertex.}
\label{fig-DMeshCCellCCoeff}
\end{figure}

\begin{figure}[htbp]
\begin{center}
\hfill
\subfigure[]{
\includegraphics[width=0.3\textwidth]{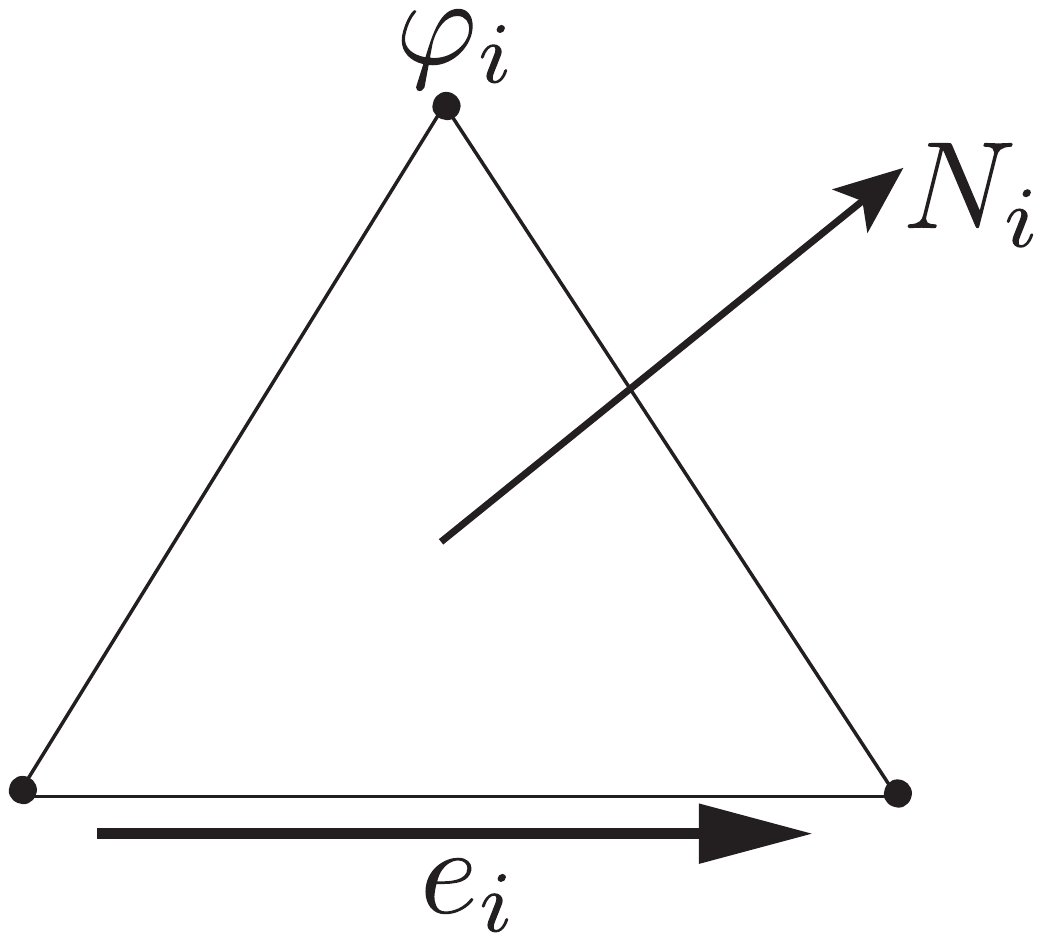}
}
\hfill
\subfigure[]{
\includegraphics[width=0.3\textwidth]{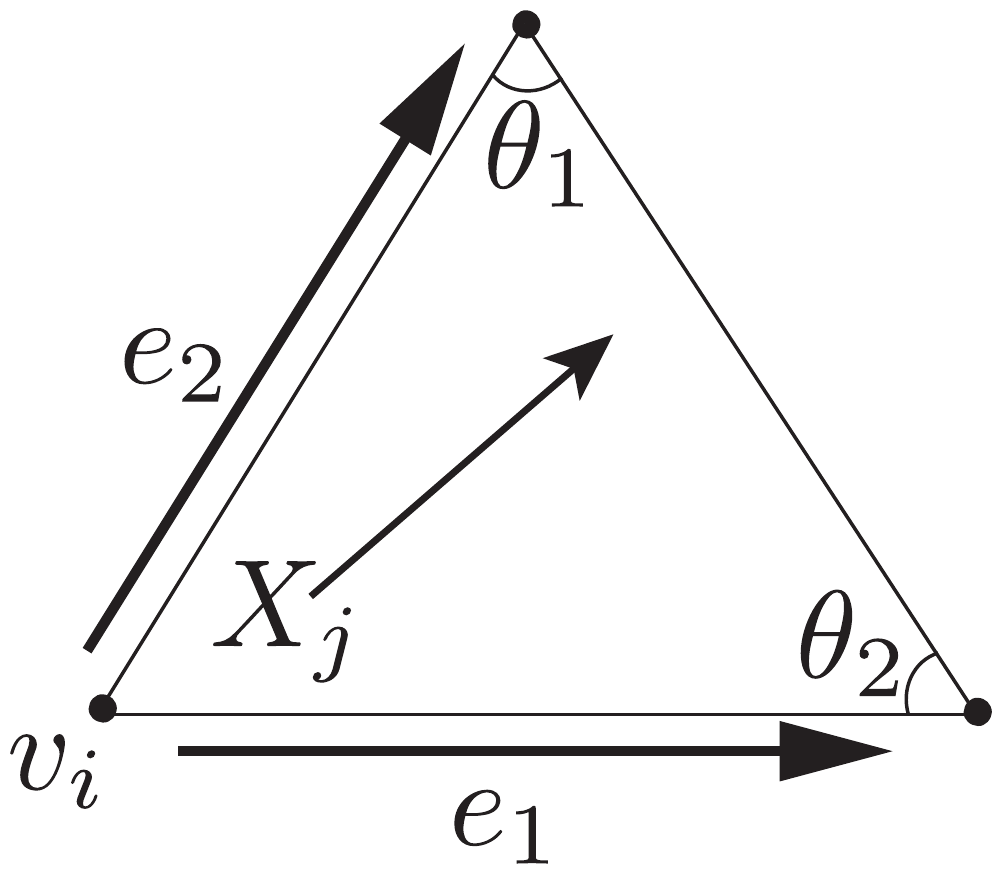}
}
\hspace*{\fill}
\caption{\label{fig-Differential-Operators} Computation of gradient and divergence within an element (i.e., a triangle). (a) gradient; (b) divergence.}
\end{center}
\end{figure}

The gradient of $\varphi$ over each triangle is constant as the function $\varphi$ is linear within the triangle. The gradient in a given triangle can be expressed as
\begin{equation}\label{equ-gradient}
\nabla \varphi (f_i)= \frac{1}{2A_{i}}\sum\limits_{j \in \Omega_i}\varphi_{j}(N_i\times e_j),
\end{equation}
where $A_{i}$ is the area of the face $f_i$, $N_i$ is its unit normal, $\Omega_i$ is the set of edge indices for face $f_i$, $e_j$ is the $j-$th edge vector (oriented counter-clockwise), and $\varphi_i$ is the opposing value of $\varphi$ as shown in Fig.~\ref{fig-Differential-Operators}.


According to the Stokes' theorem, the integral of divergence over the dual cell is equal to the outward flux along the boundary of the dual cell.  Thus the divergence operator associated with vertex $v_i$ is dicretized by dividing the outward flux by the dual cell area
\begin{equation}
\div(X) = \frac{1}{2C_i}\sum\limits_{j\in D_1(i)}\cot\theta_{j}^1(e_{j}^1\cdot X_j) + \cot\theta_{j}^2(e_{j}^2\cdot X_j),
\end{equation}
where the sum is taken over the vertex's incident triangles $f_{j}$ with a vector $X_j$, $e_{j}^1$ and $e_{j}^2$ are the two edge vectors of triangle $f_{j}$ containing vertex $v_i$, $\theta_{j}^1$ and $\theta_{j}^2$ are the opposing angles, and $C_i$ is the dual cell area for vertex $v_i$. Accordingly, the geodesic curvature value of curve $\varphi = const$ associated with the vertex $v_i$ can be computed by
\begin{equation}
\RevisionBlack{\kappa_{g}^{i} = \frac{1}{2C_i}\sum\limits_{j \in D_1(i)} \frac{\cot\theta_{j}^1 (e_{j}^2\cdot \nabla \varphi (j)) + \cot\theta_{j}^1 (e_{j}^2\cdot \nabla \varphi (j))}{\|\nabla \varphi (j)\|} .}
\end{equation}

The curvature tensor (second fundamental tensor) $T$ is defined in terms of the directional derivatives of the surface normal:
\begin{equation}
T = \left(\begin{array}{cc} D_{u}n& D_{v}n\end{array}\right) = \left( \begin{array}{cc}
\frac{\partial n}{\partial u}\cdot u & \frac{\partial n}{\partial v}\cdot u \\
\frac{\partial n}{\partial v}\cdot u & \frac{\partial n}{\partial v}\cdot v \end{array} \right),
\end{equation}
where $(u, v)$ are the directions of an orthogonal coordinate system in the tangent frame (the sign convention used here yields positive curvatures for convex surfaces with outward-facing normals). Multiplying this tensor by any vector in the tangent plane gives the derivative of the normal in that direction. Although this definition holds only for smooth surfaces, we can approximate it in the discrete case using finite difference. In this work, the curvature tensor for each face is computed by the method in~\cite{rusinkiewicz2004estimating}.

Then the whole optimization model can be formulated as
\begin{align}\label{eq:OptimizationModel}
\min_{\varphi} & \sum\limits_{j=1}^{|F|}A_{j}~ \left( \|\nabla \varphi (j)\| - \sqrt{\frac{{\kappa }_{s}^{j} + {\kappa }_{c}}{8}} \right)^2 + \lambda \left(\sum\limits_{j=1}^{|F|} A_{j} {({\kappa}_{n}^{j})}^2 + \sum\limits_{i=1}^{|V|} C_i {({\kappa}_{g}^{i})}^2\right)\nonumber\\
	\textrm{s.t.} & \quad \|\nabla\varphi(j)\| > 0.
\end{align}
where $|F|$ is the number of faces and $|V|$ denotes the number of vertices. This is a well established nonlinear least square optimization problem with inequality constraints, which can be easily solved by the interior point method~\cite{COLEMAN96,Waechter2006,CURTIS2012}. The interior point solver requires the gradient of the target function and the constraint functions. The gradient calculation boils down to computing the gradient of $\nabla\varphi$ and ${\nabla \varphi }/{\left\| \nabla \varphi  \right\|}$, which we do as follows.

As demonstrated previously, the gradient of a piecewise linear scalar function within a given triangle $f_k$ is a linear combination of constant vectors $N_k\times e_i$, and thus, the partial derivative of $\nabla\varphi (k)$ with respect to $\varphi_j$ is
\begin{equation}\label{equ:gradientOfgradient}
\RevisionBlack{\frac{\partial }{\partial \varphi_j }\nabla \varphi (k) =\frac{1}{2{{A}_{k}}}\frac{\partial }{\partial \varphi_j }\sum\limits_{i \in \Omega_k}{{{\varphi }_{i}}}(N_k \times {{e}_{i}})=\frac{1}{2{{A}_{k}}}\sum\limits_{i \in \Omega_k}{{{\delta }_{ij}}(N_k \times {{e}_{i}})},}
\end{equation}
where ${\delta _{ij}} = \left\{ {\begin{array}{*{20}{c}}
1&{i = j}\\
0&{i \ne j}
\end{array}} \right.$ is the Kronecker delta function.

As for the gradient of ${\nabla \varphi }/{\left\| \nabla \varphi  \right\|}$, it is
\begin{equation}\label{equ:partial}
\RevisionBlack{\frac{\partial }{\partial \varphi_j }\left( \frac{\nabla \varphi (k) }{\left\| \nabla \varphi (k)  \right\|} \right)=\frac{\left( \frac{\partial }{\partial \varphi_j }\nabla \varphi (k) \right)\left\| \nabla \varphi (k) \right\|-\nabla \varphi (k) \frac{{{\left( \nabla \varphi (k)  \right)}^{T}}\frac{\partial }{\partial \varphi_j }\nabla \varphi (k) }{\left\| \nabla \varphi (k)  \right\|}}{{{\left\| \nabla \varphi (k)  \right\|}^{2}}}.}
\end{equation}

The final solution to the optimization problem Equ.~\eqref{eq:OptimizationModel} would be affected by the initial value. In this work, we initialize the tool path by paths from~\cite{crane2013geodesics}.

\section{Tool path planning algorithm}
\label{sec:4}
Planning tool-path is to represent a surface with a series of curves against some error criteria (i.e., chord deviation and scallop height). We next summarize the overall process for generating such curves on a surface by the iso-level method as follows:
\begin{enumerate}
  \item Select an initial curve $C_0$ on the surface $S$ and fix its level value to zero, i.e., $l_0 = 0$. $C_0$ is a part of boundary for direction parallel tool path and the whole boundary for contour parallel tool path.
  \item Find the solution to the models Equ.~\eqref{equ:iso-scallop-energy}, Equ.~\eqref{eq:curveCurvature}, and Equ.~\eqref{eq:OptimizationModel-continue}, including meshing and numerical optimization.
  \item Select level values $\left\{ {{l}_{i}} \right\}_{i=1}^{n}$, where $ l_{1} = {{\varphi }_{min}},\ \ l_{(n)}={{\varphi }_{max}}$, with the method described in Section \ref{sec:2.1} for iso-scallop tool path and the method in Section \ref{sec:2.2} for smooth or optimal tool path. For direction parallel tool path the last tool path corresponds to ${{l}_{n}}={\varphi }_{max}$, while for contour parallel tool path the last corresponds to ${{l}_{n-1}}$. Then fastly extract iso-level curves on the triangular mesh based on the method described in Section \ref{sec:2.4}.
  \item Convert the iso-level curves on mesh which actually are polygons to surface $S$. The vertices of an iso-level curve on the mesh are either vertices of the mesh or points on edges of the mesh. For the former case, the vertices are also on $S$. For the latter case, a vertex is first proportionally mapped to the parameter domain with respect to the two ends of the edge it is on and then find its corresponding point on the surface.
  \item Greedily merge short segments of the polygons to approach the chord deviation tolerance as closely as possible. Then finally, these reduced iso-level curves (polygons) are desired tool path.
\end{enumerate}
%

\section{Experimental results}
\label{sec:5}
In this section, the proposed tool path planning method is implemented on real data. A free-form surface and a human face are chosen to illustrate the effectiveness of it, as in Fig.~\ref{fig:model}. The free-form surface is exploited to show the generation of direction parallel tool path. The human face was generated by a coordinate measuring machine. We utilize it to show the generation of contour parallel tool path.

To plan iso-level tool path, the first thing to do is to construct a proper scalar function over the surface. Since the Finite Element Method is employed to find the optimizer of the optimization models, meshing is needed. We choose the element to be triangular. Fig.~\ref{fig:surface-mesh} shows the meshing results of the free-form surface and Fig.~\ref{fig:face-mesh} shows that of the human face. The optimal scalar functions are illustrated in Fig.~\ref{fig:surface-function} and Fig.~\ref{fig:face-function} by varying color. Fig. \ref{fig:surface-function} shows the scalar function of the free-form surface for generating direction parallel tool path and Fig. \ref{fig:face-function} shows that of the human face for generating contour parallel tool path. And the varying from blue to red represents the rising of level value.

As the optimal scalar functions have been constructed for both surfaces, tool path that is optimal with respect to iso-scallop and smoothness can be generated. A ball-end cutter with radius 4mm is chosen to show the path generation so that tool orientation doesn't matter. The limited scallop height is 1mm and chord deviation is 0.01mm. In order to clearly show tool paths, the error criterion (i.e., scallop height) is set to be much greater than those in real cases. Fig.~\ref{fig:surface-path} shows the optimal direction parallel paths on the free-form surface and Fig.~\ref{fig:face-path} shows corresponding result of contour parallel tool path on the human face. Their weights are both $ \lambda  = 1$.

We next show some comparisons and analyses of the generated tool paths. According to the demonstration of \cite{kim2002machining}, the contour parallel tool path will be emphasized. Fig.~\ref{fig:From-smooth-isoscallop} shows the tool paths form smooth to iso-scallop generated by the proposed method. Fig. \ref{fig:high-weight-path} shows the smooth contour parallel tool path generated by the proposed method with $\lambda = 10$, Fig. \ref{fig:mid-weight-path} shows the optimal tool path with $\lambda = 1$, and Fig. \ref{fig:low-weight-path} shows the iso-scallop tool path with $\lambda = 0$. As described in above sections, the iso-scallop condition  Equ.~\eqref{equ:5} characterizes the overlapping between neighbor paths. Therefore, to analyze the overlapping of the generated tool paths, we conduct statistics on the relative deviation w.r.t. the iso-scallop condition along the paths. It is computed by
\begin{equation}\label{}
\Delta =\left| \frac{\left\| \nabla \varphi  \right\|-\sqrt{\frac{{{\kappa }_{s}}+{{\kappa }_{c}}}{8}}}{\sqrt{\frac{{{\kappa }_{s}}+{{\kappa }_{c}}}{8}}} \right|=\left| 1-\frac{\left\| \nabla \varphi  \right\|}{\sqrt{\frac{{{\kappa }_{s}}+{{\kappa }_{c}}}{8}}} \right|.
\end{equation}
And the statistics results are depicted in Fig.~\ref{width-histogram-high-weight}, \ref{width-histogram-mid-weight}, \ref{width-histogram-low-weight}. As the figures show, for the iso-scallop tool path, the relative deviation are all less than 5\%, and centered around 1\%. For the optimal tool path, we could find the ratio move to the greater side, as imagined, and there are a few points which are much greater than the rest. Most of these points are located in the corner parts of the tool path. And for the smooth tool path, its overlapping is much more obvious and there are about 2\% of points whose ratio are greater than 10\%. But the losing of iso-scallop condition brings smoothness to the tool paths, which is shown in Fig.~\ref{fig:curvature-histogram-high-weight}, \ref{curvature-fig:histogram-mid-weight}, \ref{curvature-fig:histogram-low-weight}. In conclusion, the optimal tool path tries to find a balance between the overlapping and smoothness. We also compare the optimal tool path with the Laplacian based one in Fig.~\ref{fig:comparison-laplace}. Although the Laplacian based tool path is obviously smooth than the optimal one, from the overlapping analysis figures, i.e., Fig.~\ref{fig:width-histogram-laplace}, \ref{fig:width-histogram-optimal}, we can find that it is much more severely overlapped for neighbor paths.

\begin{figure*}[htbp]
\centering
\hfill
\subfigure[]{
\begin{minipage}[b]{0.2\textwidth}
\includegraphics[width=1.8\textwidth, origin=bl]{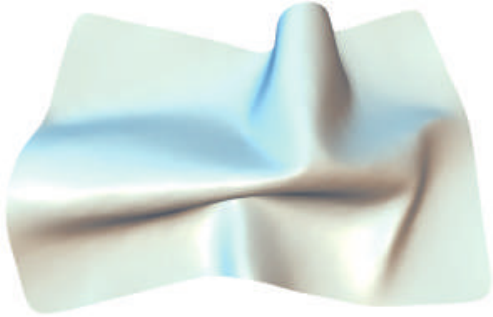}
\end{minipage}\label{fig:surface-model}
}
\hfill \hfill
\subfigure[]{
\begin{minipage}[b]{0.2\textwidth}
\includegraphics[width=0.85\textwidth]{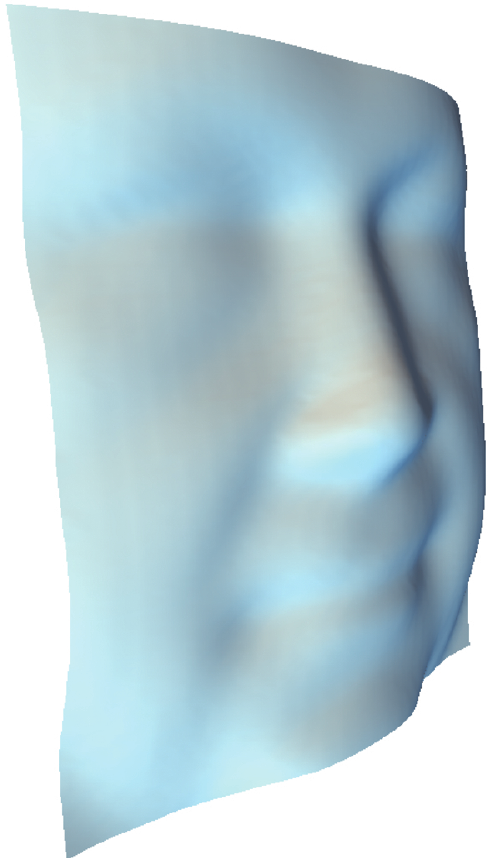}
\end{minipage}\label{fig:face-model}
}
\hspace*{\fill}
\caption{Tested models. (a) free-form surface; (b) human face.}\label{fig:model}
\end{figure*}

\begin{figure*}[htbp]
\centering
\hfill
\subfigure[]{
\begin{minipage}[b]{0.2\textwidth}
\includegraphics[width=1.4\textwidth]{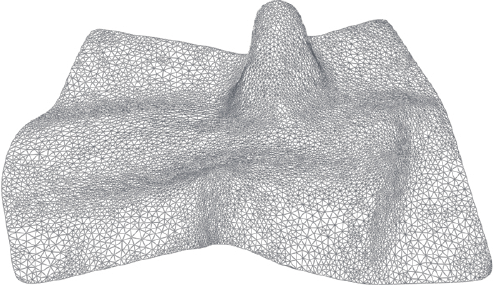}
\end{minipage}\label{fig:surface-mesh}
}
\hfill
\subfigure[]{
\begin{minipage}[b]{0.2\textwidth}
\includegraphics[width=1.4\textwidth]{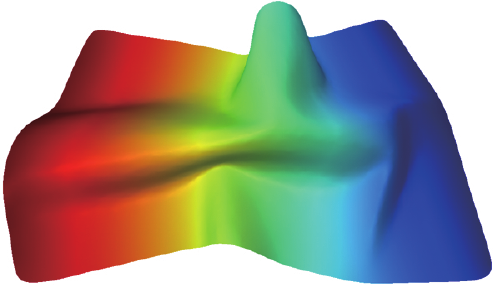}
\end{minipage}\label{fig:surface-function}
}
\hfill
\subfigure[]{
\begin{minipage}[b]{0.2\textwidth}
\includegraphics[width=1.4\textwidth]{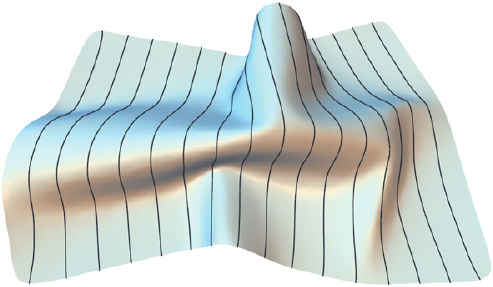}
\end{minipage}\label{fig:surface-path}
}
\hspace*{\fill}
\caption{Direction parallel tool path for the free-form surface. (a) meshing result; (b) optimal scalar function; (c) the generated tool path.}\label{fig:resultDirection}
\end{figure*}

\begin{figure*}[htbp]
\centering
\hfill
\subfigure[]{
\begin{minipage}[b]{0.2\textwidth}
\includegraphics[width=1.2\textwidth, origin=bl]{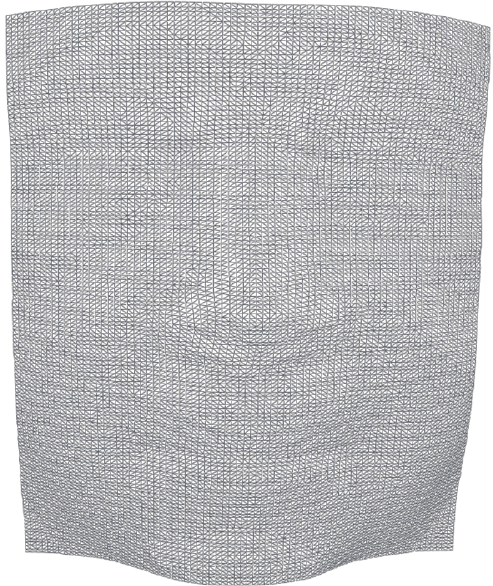}
\end{minipage}\label{fig:face-mesh}
}
\hfill
\subfigure[ ]{
\begin{minipage}[b]{0.2\textwidth}
\includegraphics[width=1.2\textwidth]{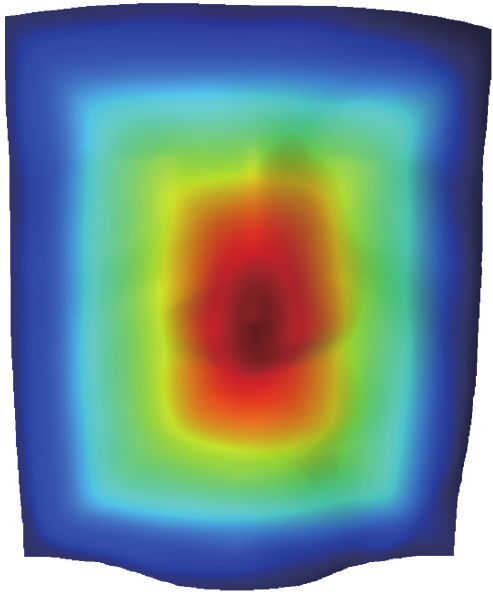}
\end{minipage}\label{fig:face-function}
}
\hfill
\subfigure[ ]{
\begin{minipage}[b]{0.2\textwidth}
\includegraphics[width=1.2\textwidth]{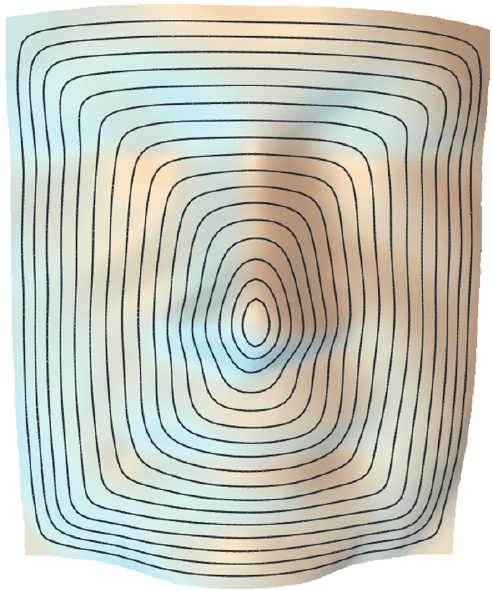}
\end{minipage}\label{fig:face-path}
}
\hspace*{\fill}
\caption{Contour parallel tool path for the face model. (a) meshing result; (b) optimal scalar function; (c) the generated tool path.}\label{fig:resultContour}
\end{figure*}

\begin{figure*}[!htb]
\centering
\hfill
\subfigure[ ]{
\begin{minipage}[b]{0.2\textwidth}
\includegraphics[width=1.2\textwidth, origin=bl]{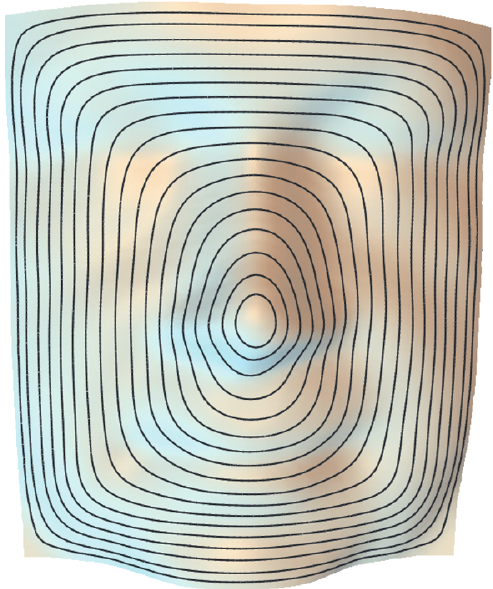}
\end{minipage}\label{fig:high-weight-path}
}
\hfill
\subfigure[ ]{
\begin{minipage}[b]{0.2\textwidth}
\includegraphics[width=1.2\textwidth]{face-smooth-1.pdf}
\end{minipage}\label{fig:mid-weight-path}
}
\hfill
\subfigure[ ]{
\begin{minipage}[b]{0.2\textwidth}
\includegraphics[width=1.2\textwidth]{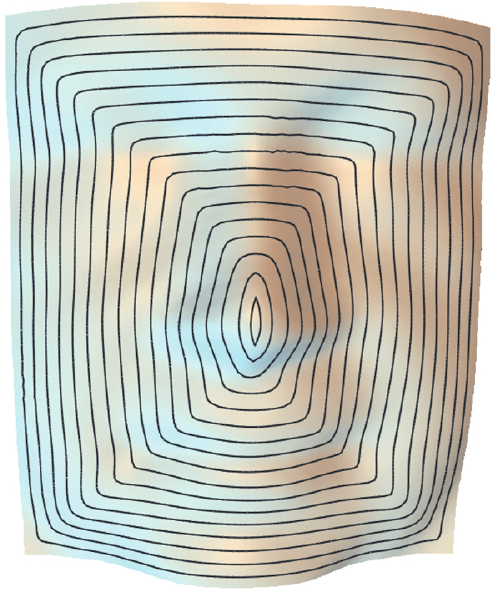}
\end{minipage}\label{fig:low-weight-path}
}
\hspace*{\fill} \\

\hfill
\subfigure[]{
\begin{minipage}[b]{0.2\textwidth}
\includegraphics[width=1.6\textwidth, origin=bl]{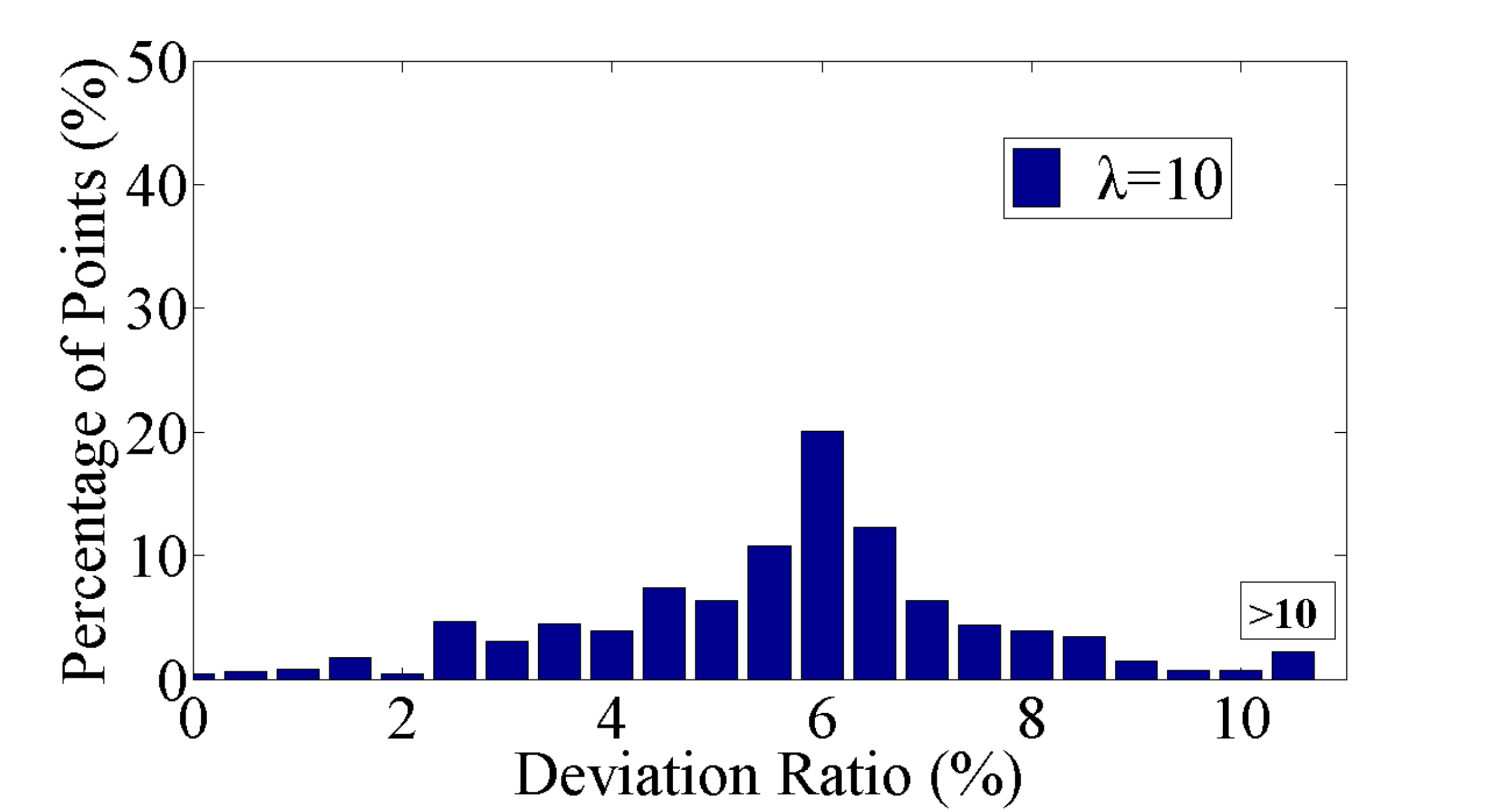}
\end{minipage}\label{width-histogram-high-weight}
}
\hfill
\subfigure[]{
\begin{minipage}[b]{0.2\textwidth}
\includegraphics[width=1.6\textwidth]{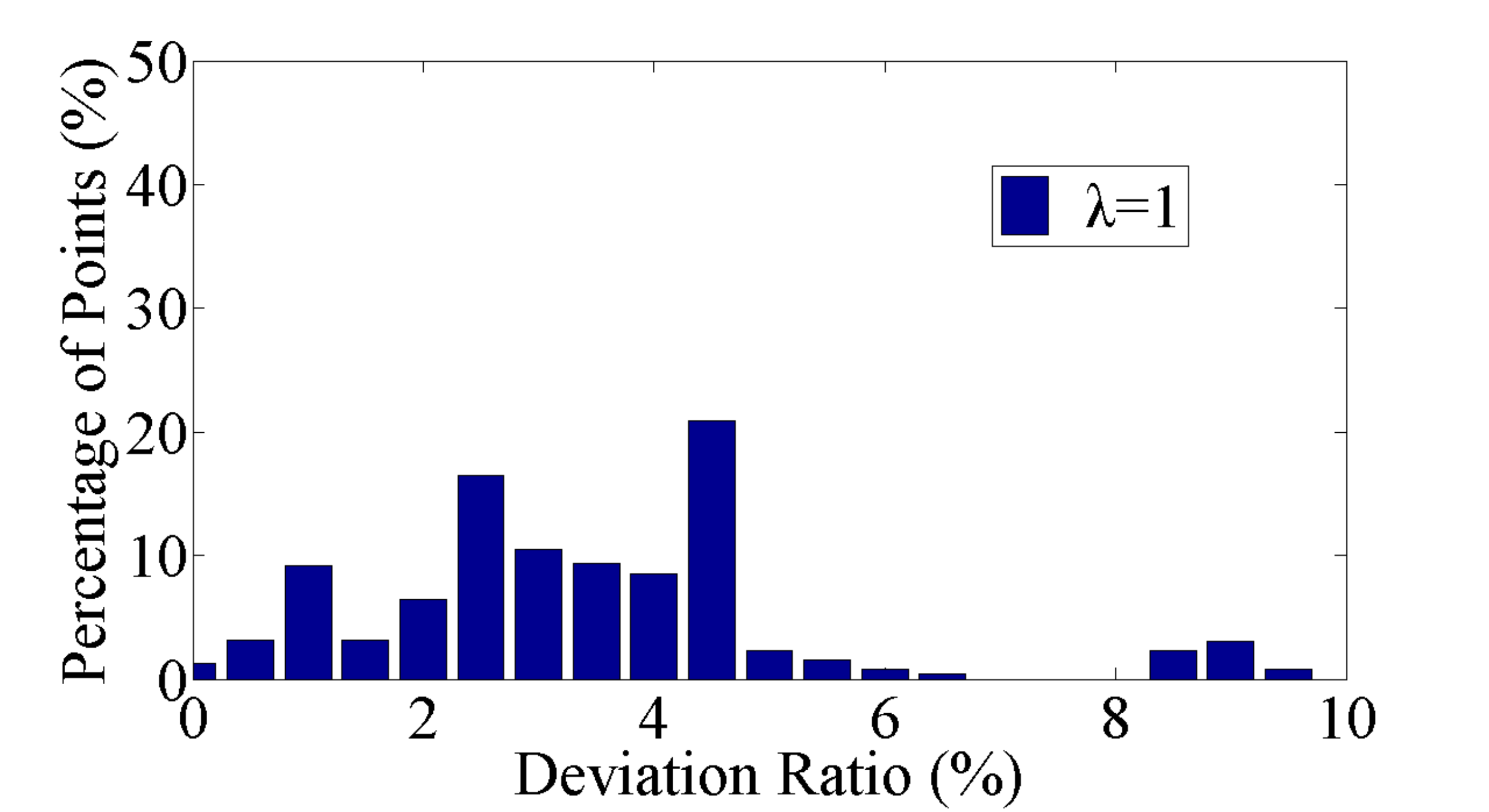}
\end{minipage}\label{width-histogram-mid-weight}
}
\hfill
\subfigure[]{
\begin{minipage}[b]{0.2\textwidth}
\includegraphics[width=1.6\textwidth]{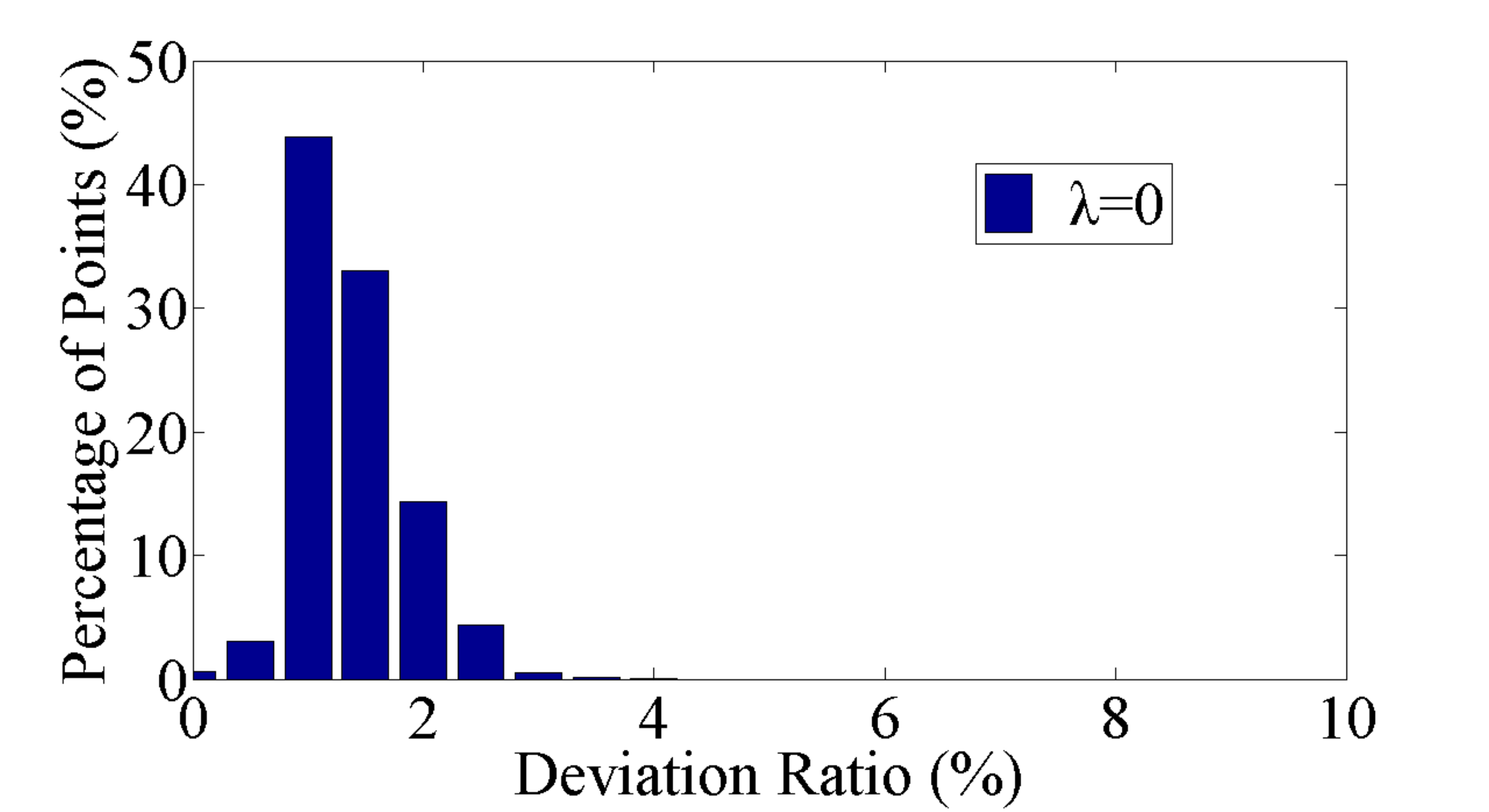}
\end{minipage}\label{width-histogram-low-weight}
}
\hfill
\hspace*{\fill} \\

\hfill
\subfigure[]{
\begin{minipage}[b]{0.2\textwidth}
\includegraphics[width=1.6\textwidth, origin=bl]{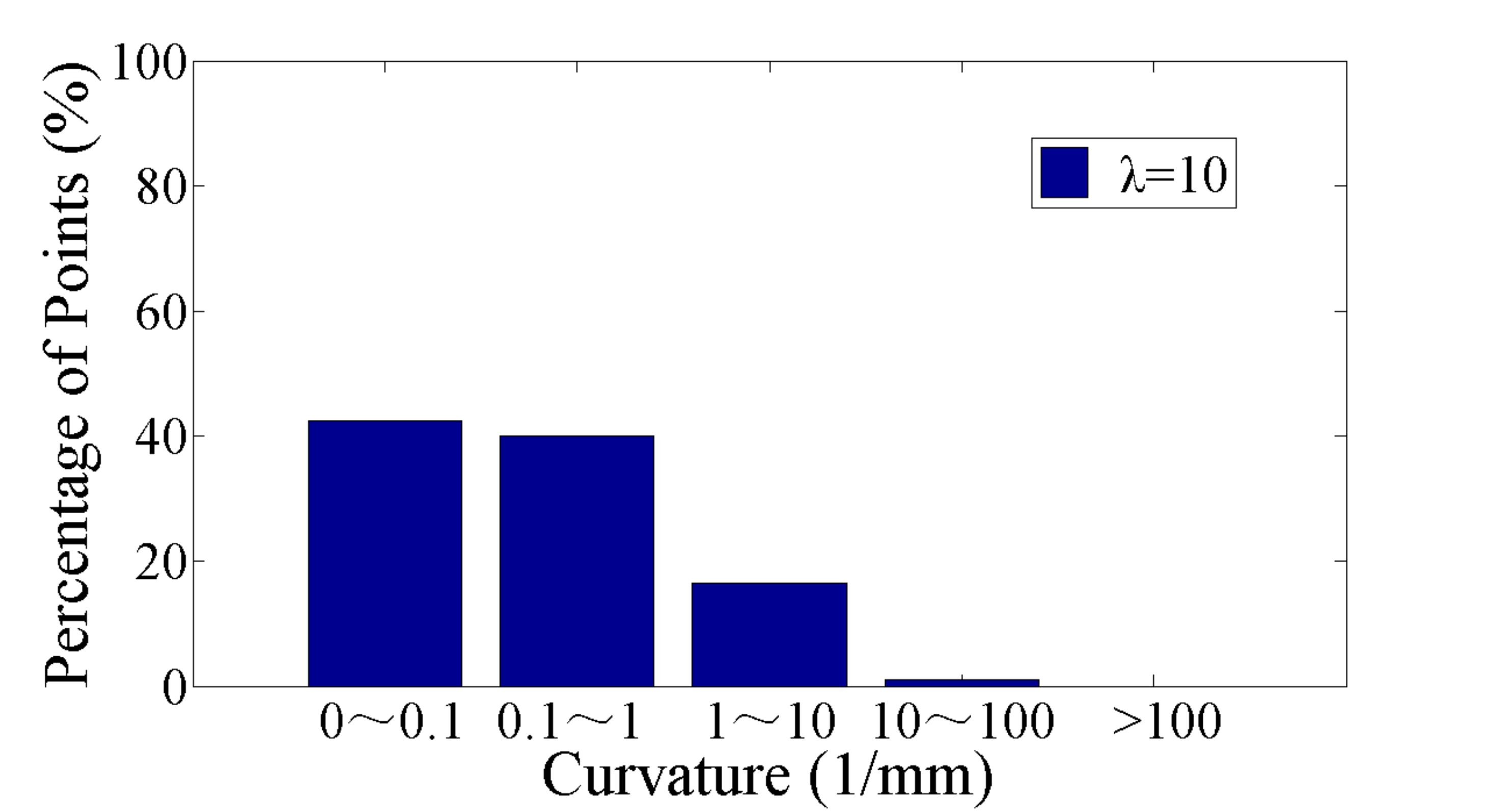}
\end{minipage}\label{fig:curvature-histogram-high-weight}
}
\hfill
\subfigure[]{
\begin{minipage}[b]{0.2\textwidth}
\includegraphics[width=1.6\textwidth]{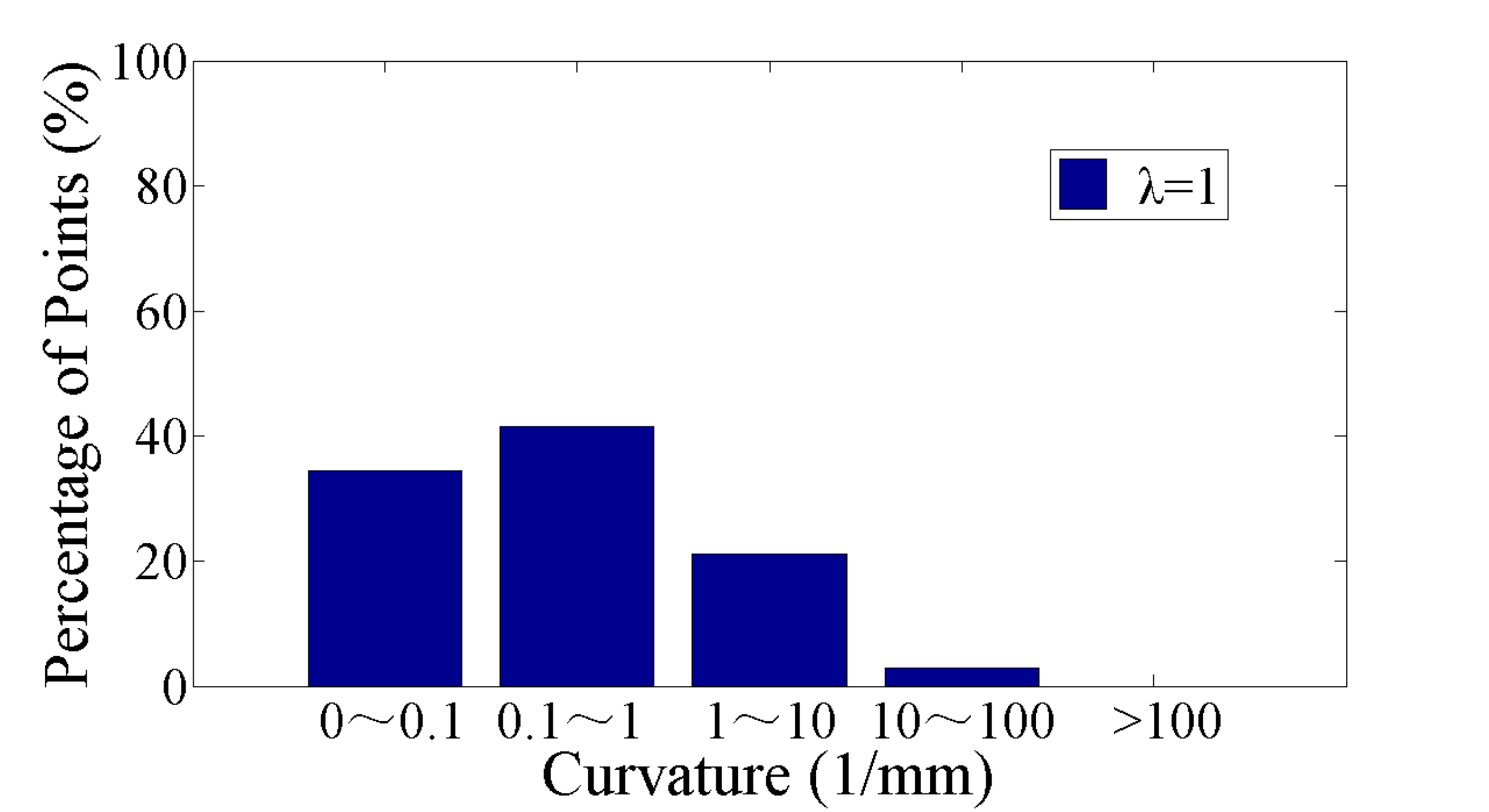}
\end{minipage}\label{curvature-fig:histogram-mid-weight}
}
\hfill
\subfigure[]{
\begin{minipage}[b]{0.2\textwidth}
\includegraphics[width=1.6\textwidth]{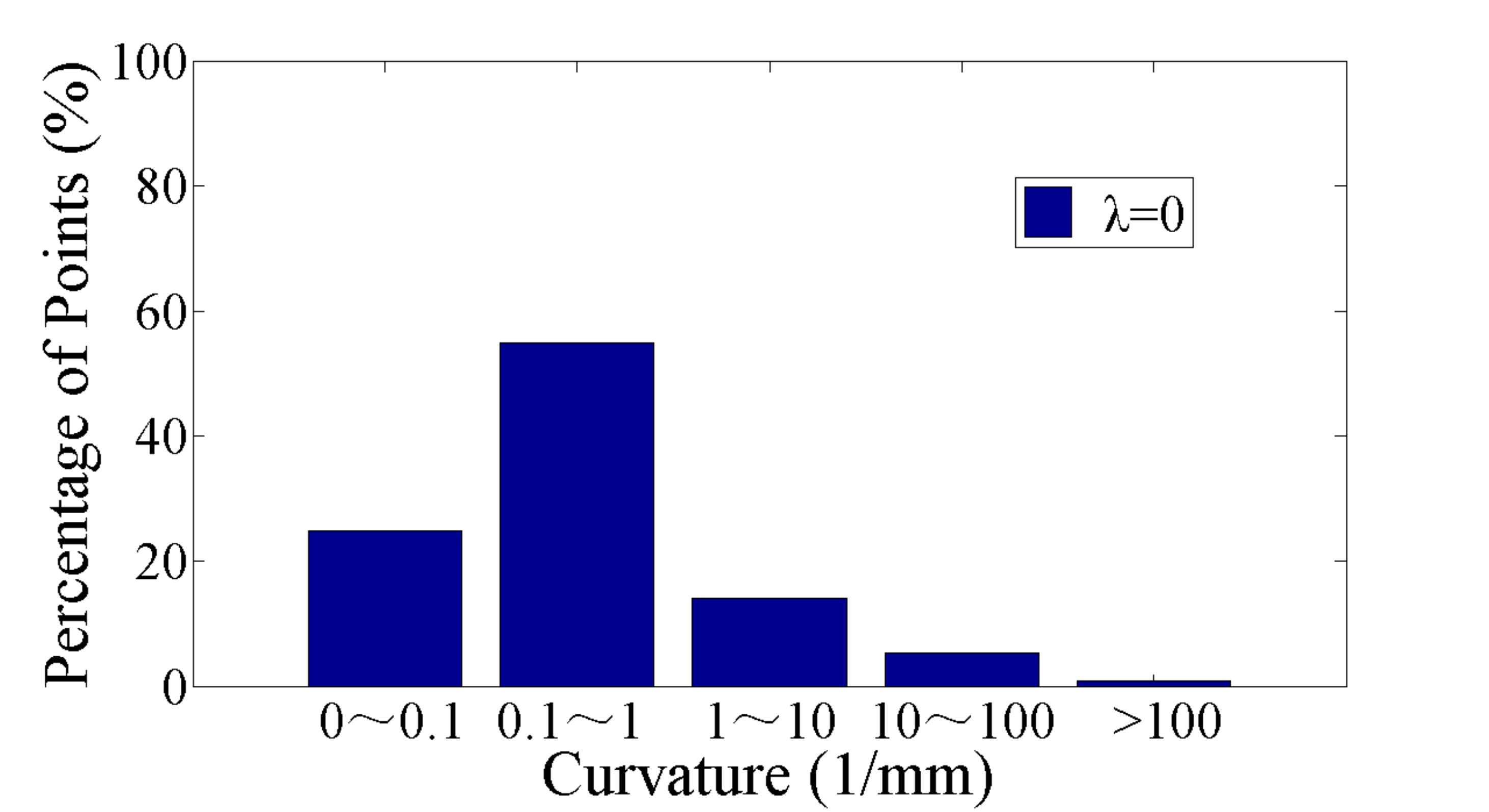}
\end{minipage}\label{curvature-fig:histogram-low-weight}
}
\hfill
\hspace*{\fill}
\caption{Tool paths from smooth to iso-scallop and their analyses. (a) smooth tool path; (b) optimal tool path; (c) iso-scallop tool path; (d) overlapping analysis for smooth tool path; (e) overlapping analysis for optimal tool path; (f) overlapping analysis for iso-scallop tool path; (g) curvature analysis for smooth tool path; (h) curvature analysis for optimal tool path; (i) curvature analysis for iso-scallop tool path.} \label{fig:From-smooth-isoscallop}
\end{figure*}

\begin{figure*}[!htb]
\centering
\hfill
\subfigure[ ]{
\begin{minipage}[b]{0.2\textwidth}
\includegraphics[width=1.2\textwidth, origin=bl]{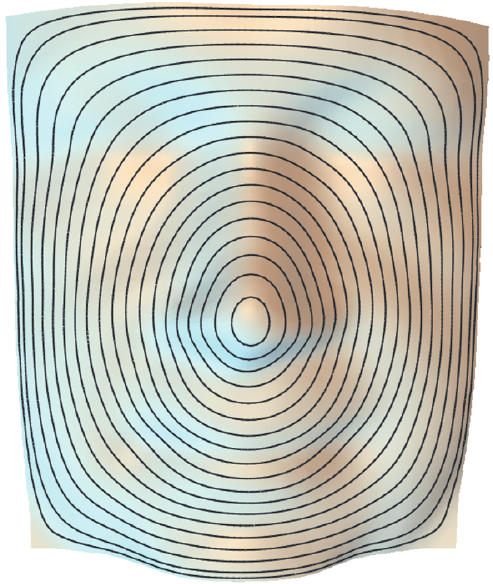}
\end{minipage}\label{fig:laplace-path}
}
\hfill
\subfigure[ ]{
\begin{minipage}[b]{0.2\textwidth}
\includegraphics[width=1.2\textwidth]{face-smooth-1.pdf}
\end{minipage}\label{fig:optimal-path}
}
\hspace*{\fill} \\
\hfill
\subfigure[]{
\begin{minipage}[b]{0.2\textwidth}
\includegraphics[width=2\textwidth, origin=bl]{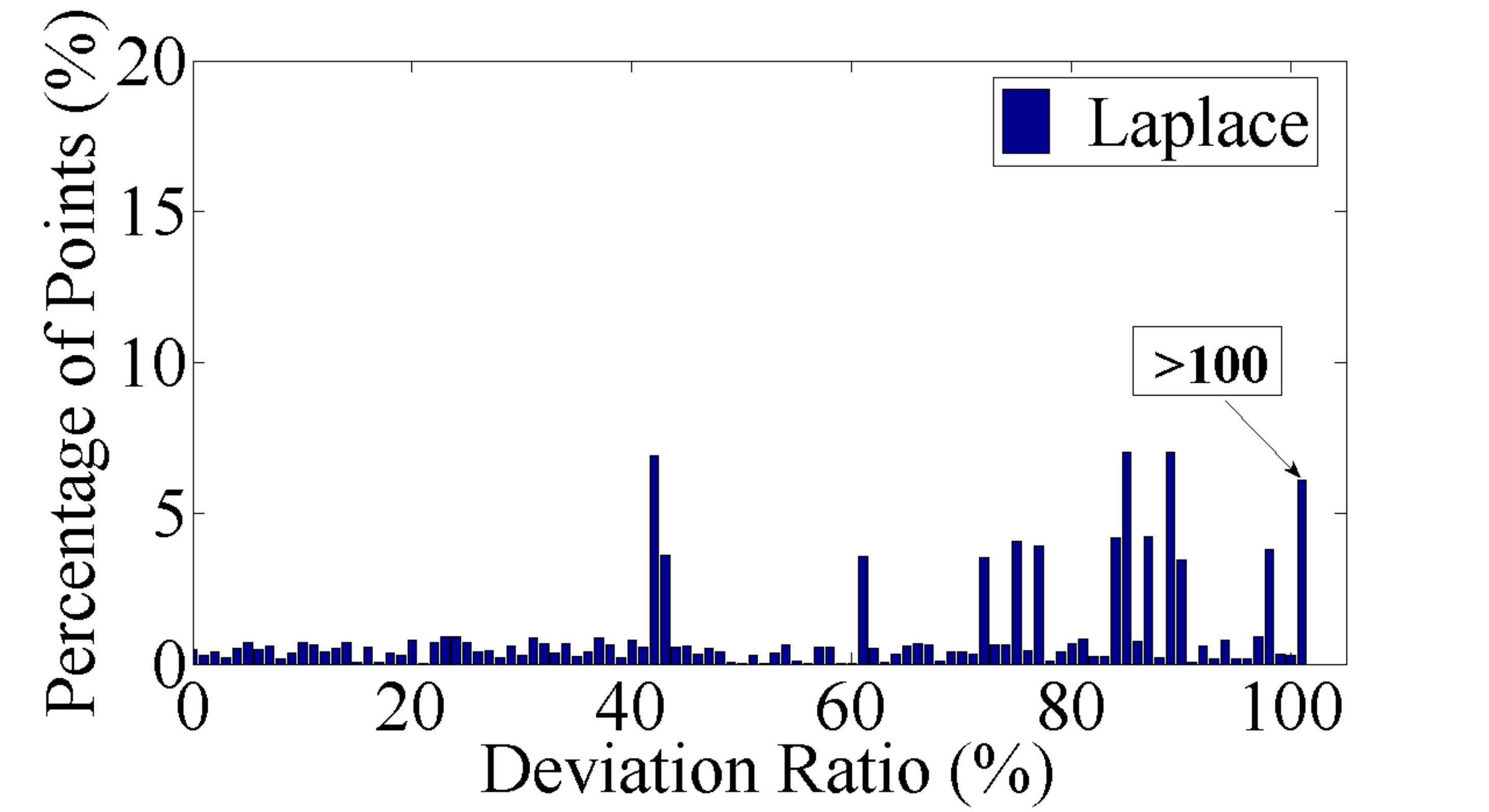}
\end{minipage}\label{fig:width-histogram-laplace}
}
\hfill
\subfigure[]{
\begin{minipage}[b]{0.2\textwidth}
\includegraphics[width=2\textwidth]{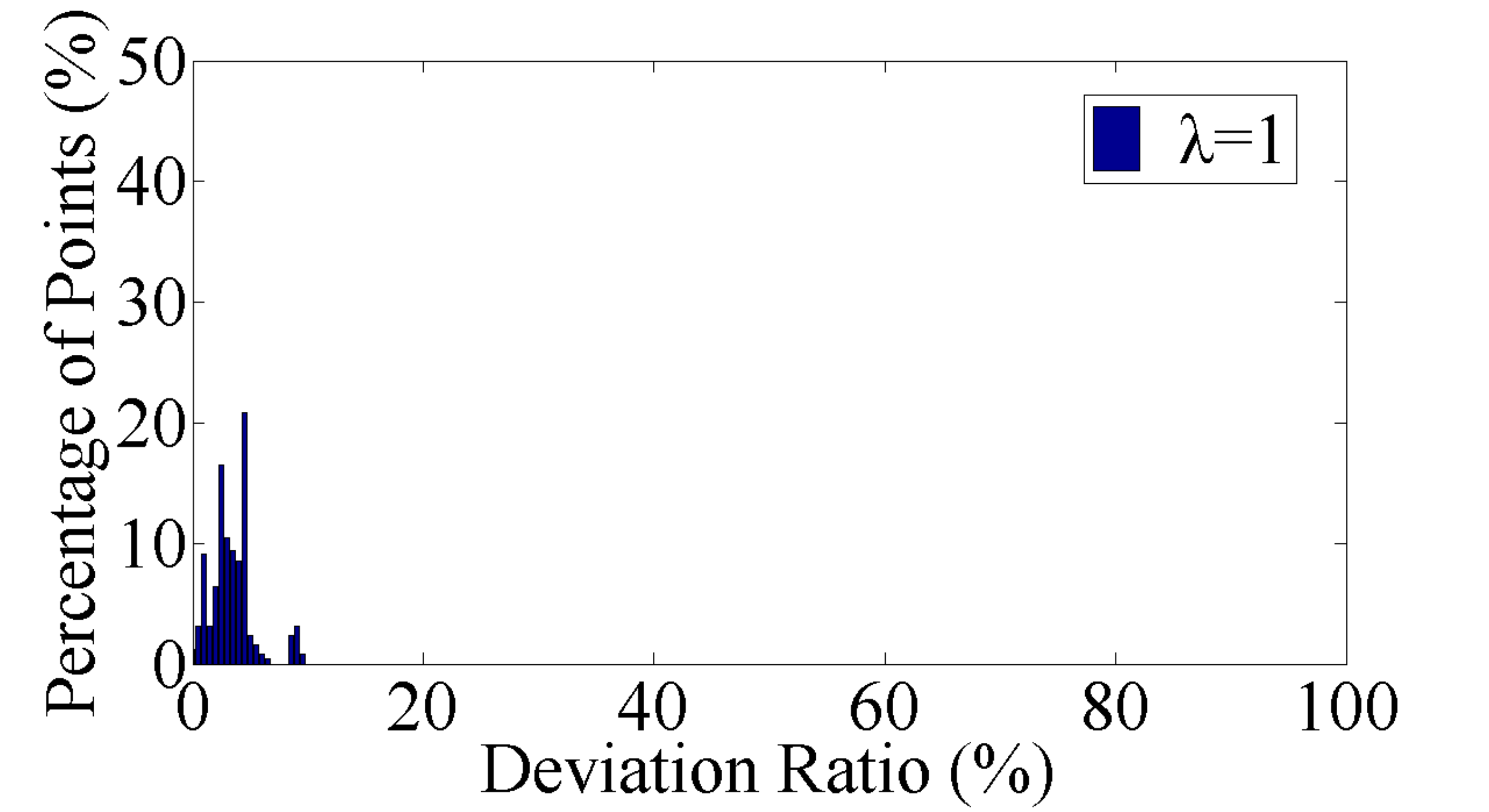}
\end{minipage}\label{fig:width-histogram-optimal}
}
\hfill
\hspace*{\fill}
\caption{Comparison of Laplacian based tool path and the optimal tool path. (a) Laplacian based tool path; (b) optimal tool path; (c) overlapping analysis for Laplacian based tool path; (d) overlapping analysis for optimal tool path.} \label{fig:comparison-laplace}
\end{figure*}




\section{Conclusion}
\label{sec:6}
In this paper, a new framework of tool path planning is proposed. The novelty of our method is that it allows several objectives to be considered in a unified framework and thus making global optimization of tool paths possible. Moreover, the scalar function only has to be constructed once, then it can be utilized to generate tool paths for machining from rough to fine. The proposed framework is applied to find an optimal tool path that takes smoothness and iso-scallop requirements into consideration simultaneously. Equ.~\eqref{equ:5} for controlling interval between neighbor iso-level curves and Equ.~\eqref{equ:smooth-term} for measuring curvature of an iso-level curve are derived to lay a foundation for the formulation of optimization models.

It is likely that this theory has further potential in planning other optimal tool path, and the derived formulas can also be directly applied to level set based tool path planning methods, e.g., \cite{dhanik2010contour}.


\section{Acknowledgement}
\label{sec:7}
The authors are grateful for the support provided by National Key Basic Research Project of China (No.~2011CB302400), National Natural Science Foundation of China (No.~61303148, 50975495), Ph.D. Programs Foundation of Ministry of Education of China (No.~20133402120002), and Swiss National Science Foundation (No. 200021\_137626).




\bibliographystyle{elsarticle-num}
\bibliography{CADbibliography}



%
%
%
\end{document}